\documentclass[11pt, letterpaper]{article}
\usepackage{microtype}
\usepackage{amssymb}
\usepackage{fullpage}
\usepackage[noadjust]{cite}
\usepackage[noend]{algpseudocode}
\usepackage{times}
\newcommand*\samethanks[1][\value{footnote}]{\footnotemark[#1]}

\usepackage{lipsum}
\let\svfootnoterule\footnoterule
\renewcommand\footnoterule{\vfill\svfootnoterule}
\title{Space Lower Bounds for Itemset Frequency Sketches\thanks{This paper supersedes an earlier manuscript of the same title by the first three authors.}}

\author{Edo Liberty\thanks{Yahoo Labs}
\and
Michael Mitzenmacher\thanks{Harvard University, School of Engineering and Applied Sciences.}
\and
Justin Thaler\samethanks[2]
\and
Jonathan Ullman\thanks{Northeastern University, College of Computer and Information Sciences.}
}

\date{}

\usepackage{amsfonts,amsmath,color,float,graphicx,verbatim,url, amsthm}

\usepackage{algorithm}

\usepackage{enumerate}

\newtheorem{theorem}{Theorem}
\newtheorem{lemma}[theorem]{Lemma}

\newtheorem{definition}[theorem]{Definition}
\newtheorem{fact}[theorem]{Fact}

\newcommand{\Enc}{\mathrm{Enc}}

\newcommand{\releasedb}{\textsc{release--db}}
\newcommand{\releaseanswers}{\textsc{release--answers}}
\newcommand{\subsample}{\textsc{subsample}}
\newcommand{\database}{\mathcal{D}}
\newcommand{\eps}{\epsilon}
\newcommand{\rows}{n}

\newcommand{\cols}{d}
\newcommand{\poly}{\operatorname{poly}}
\newcommand{\bits}{\{0,1\}}
\newcommand{\itemset}{T}

\newcommand{\sumalg}{\mathcal{S}}
\newcommand{\recalg}{\mathcal{Q}}

\newcommand{\alg}{\mathcal{A}}
\newcommand{\eat}[1]{}

\usepackage{soul}

\usepackage{color}

\newcommand{\IFE}{For-All-Itemset-Frequency-Estimator}
\newcommand{\IFI}{For-All-Itemset-Frequency-Indicator}
\newcommand{\SIFE}{For-Each-Itemset-Frequency-Estimator}
\newcommand{\SIFI}{For-Each-Itemset-Frequency-Indicator}

\newcommand{\bfJ}{\mathbf{J}}

\begin{document}
\maketitle
\begin{abstract}
Given a database, computing the fraction of rows that contain a query
itemset or determining whether this fraction is above some threshold
are fundamental operations in data mining.  
 A uniform sample of rows is a good
sketch of the database in the sense that all sufficiently frequent
itemsets and their approximate frequencies are recoverable from the
sample, and the sketch size is independent of the number of rows in
the original database. For many seemingly similar problems there are
better sketching algorithms than uniform sampling.  In this 
paper we show that for itemset frequency sketching this is not the case.
That is, we prove that there exist classes of databases for which uniform
sampling is a space optimal sketch for approximate
itemset frequency analysis, up to constant or iterated-logarithmic factors.
\end{abstract}



\section{Introduction}
Identifying frequent itemsets is one of the most basic and well-studied problems in data mining. 
Formally, we are given a binary database $\database \in \left(\bits^{\cols}\right)^{\rows}$ consisting of $\rows$ rows and $\cols$ columns, or attributes.\footnote{Throughout, we use the terms \emph{attributes} and \emph{items} interchangeably. While attributes may be non-binary in many applications, any attribute with $m$ possible values can be decomposed into $2 \lceil \log m \rceil$ binary attributes, using two binary attributes to mark whether the value
is 0 or 1 in the $i$th bit.  We therefore focus on the binary case.}
An \emph{itemset} $\itemset \subseteq [\cols]$ is a subset of the attributes, and the \emph{frequency} $f_{\itemset}$ of $\itemset$ is the fraction of rows of $\database$ that have a 1 in all columns in $\itemset$.

Computing itemset frequencies is a central primitive that can be used to solve the following problems (and countless others): given a large corpus of text
files, compute the number of documents containing a specific search
query; given user records, compute the fraction of users who belong to
a specific demographic; given event logs, compute sets of events that
are observed together; given shopping cart data, identify bundles of items that are
frequently bought together. 

In many settings, an approximation of $f_T$, as opposed to an exact result,
suffices. Alternatively, in some settings it suffices to recover a
single bit indicating whether or not $f_\itemset \ge \eps$ for some
user defined threshold $\eps$; such frequent itemsets may require
additional study or processing.  It is easy to show that uniformly
sampling $\poly(d/\eps)$ rows from $\database$ and
computing the approximate frequencies on the sample
$\sumalg(\database)$ provides good approximations to $f_\itemset$ up
to additive error $\eps$.  Our main contribution is to
provide lower bounds establishing that uniform sampling is an essentially
optimal \emph{sketch}, in terms of the space/accuracy tradeoff, for many parameter
regimes.  Here, a sketch $\sumalg(\database)$ of a database
is a bit string that enables recovery of accurate approximations to itemset frequencies. 

Note that, unlike a row sample, in general a sketch is not
limited to containing a subset of the database rows. Our lower bounds
hold for any summary data structure and recovery algorithm that
constitute a valid sketch.

\subsection{Motivation}
\subsubsection{The Case Against Computing Frequent Itemsets Exactly}
\label{sec:caseagainst}
If the task is only to identify frequent itemsets ($f_T \ge \eps$ for some $\eps$), it is natural to ask whether we can compute all $\eps$-frequent itemsets and store only those.
Assuming that only a small fraction of itemsets are $\eps$-frequent, 
this will result in significant space saving relative to na\"ive solutions.
The extensive literature on exact mining of frequent itemsets 
dates back to work of Agrawal et al. \cite{agrawal1993}, whose motivation
stemmed from the field of market basket analysis. As the search space
for frequent itemsets is exponentially large (i.e., size $2^{\cols}$), substantial effort was devoted
to developing algorithms that rapidly prune the search space and minimize 
the number of scans through the database. 
While the algorithms developed 
in this literature offer substantial concrete speedups over naive approaches to frequent
itemset mining, they may still take time $2^{\Omega(d)}$, simply because there may
be these many frequent itemsets.  For example, if there is a frequent itemset of cardinality $d/10$, each of its $2^{d/10}$ subsets is also frequent.
Motivated by this simple observation, there is now an extensive literature on condensed or non-redundant representations of exact frequent itemsets. 
Reporting only \emph{maximal} and \emph{closed} frequent itemsets often works well in practice, 
but it still requires exponential size in the worst case (see the survey \cite{condensedsurvey}).

Irrespective of space complexity, the above methods face computational challenges.
Yang \cite{Yang2004} determined that counting all frequent itemsets is \#P-complete, 
and a bottleneck for enumeration is that the number of frequent itemsets can be exponentially large.
Hamilton et al. \cite{Hamilton2006} provide further hardness results based on parametrized complexity.
Here we observe that finding even a single frequent itemset of approximately maximal size is NP-hard.
(The authors of \cite{Li05} noticed this connection as well but did not mention approximation-hardness.)

Consider the bipartite graph containing $n$ nodes (rows) on one side and $d$ nodes (attributes) on the other. An edge exists between the two sides if and only if the row contains the attribute with value $1$.
Assume there exists a frequent itemset of cardinality $\eps n$ and frequency $\eps$. 
This itemset induces a balanced complete bipartite subgraph with $\eps n$ nodes on each side. Likewise, any balanced complete bipartite subgraph with $\eps n$ 
nodes per side implies the existence of an itemset of cardinality $\eps n$ and frequency $\eps$. 
Finding the maximal balanced complete bipartite subgraph is NP-hard, and approximating it requires superpolynomial time assuming that SAT cannot be solved in subexponential time \cite{Feige04hardnessof}. 
Hence, finding an itemset of approximately maximal frequency requires superpolynomial time under the same assumption.

\subsubsection{The Case for Itemset Sketches}
Determining the smallest possible itemset sketches (as defined in \S~\ref{sec3}) is of interest in several data analysis settings.

\medskip
\noindent \textbf{Interactive Knowledge Discovery.} Knowledge discovery in databases is often an interactive process: an analyst poses a sequence of queries to the dataset, with later queries depending on
the answers to earlier ones \cite{mannila}.
For large databases, it may be inefficient or even infeasible to 
reread the entire dataset every time a query is posed. Instead, a user can 
keep around an itemset sketch only; this sketch will be much smaller than the original
database, while still providing fast and accurate answers to itemset
frequency queries.

\medskip
\noindent \textbf{Efficient Data Release.}  Itemset oracles capture a central problem in \emph{data release}. 
Here, a data curator (such as a government agency like the US Census Bureau)
wants to make a dataset publicly available. Due to their utility and ease of interpretation,
the data format of choice in these settings is typically \emph{marginal contingency tables} (marginal tables for short). For any itemset $\itemset \subseteq [\cols]$ with $|\itemset| = k$, 
the marginal table corresponding to $\itemset$ has $2^k$
entries, one for each possible setting of the attributes in $\itemset$; each entry
counts how many rows in the database are consistent with
the corresponding setting of the $k$ attributes.
Notice that marginal tables are
essentially just a list of itemset frequencies 
for $\database$.\footnote{More precisely,
itemset frequency queries are equivalent to \emph{monotone conjunction} queries on a database, while marginal tables are equivalent to general (non-monotone) conjunction queries.}

However, marginal tables can be 
extremely large (as any $k$-attribute marginal table has $2^k$ entries and there $\binom{d}{k}$ such tables),
and each released table may be downloaded by an enormous number of users. 
Rather than releasing marginal tables in their entirety, the data curator can instead 
choose to release an itemset summary. This summary can be much smaller than even a single $k$-attribute marginal table, while still permitting any user to obtain fast and accurate estimates for the frequency of any $k$-attribute marginal query. 

\medskip
\noindent \textbf{Mitigating Runtime Bottlenecks.}
While the use of itemset sketches cannot circumvent
the hardness results discussed in Section \ref{sec:caseagainst}, in many 
settings the empirical runtime bottleneck is the number of scans through
the database, rather than the total runtime of the algorithm. 
The use of itemset sketches eliminates the need for the user to repeatedly
scan or even keep a copy of the original database.  The user
can instead run a computationally intensive algorithm on the sketch
to solve (natural approximation variants) of the hard decision or search problems. 
Indeed, there has been considerable work in the data mining
community devoted to bounding the magnitude of errors that build up as a result
of using approximate itemset frequency information 
when performing more complicated data mining tasks, such as rule identification
\cite{mannila}.

\subsection{Other Prior Work}
The idea of producing condensed representations of approximate frequent itemsets
is not new. Most relevant to our work, an influential paper by Mannila and Toivonen defined the notion of an 
\emph{$\eps$-adequate representation} of any class of queries \cite{mannila}.
Our \IFE\ sketching task essentially asks
for an $\eps$-adequate representation for the class of all itemset frequency queries. 
Mannila and Toivonen analyzed the magnitude of errors that build up
when using $\eps$-adequate representations to 
perform more complicated data mining tasks, such as rule identification.
Subsequent work by Boulicaut et al. \cite{freesets} presented algorithms yielding
$\eps$-adequate representations for the class of all itemset queries, while Pei et al. \cite{peietal}
gave algorithms for approximating the frequency of all \emph{frequent} itemsets to error $\eps$.
Unlike the trivial algorithms that we describe in Section~\ref{naiveAlgorithms}, 
the algorithms presented in \cite{mannila, freesets, peietal} take exponential time in the worst case,
and do not come with worst-case guarantees on the size of the output.

Streaming algorithms for both exact and approximate variants of frequent itemset mining have also been extensively studied, in a line of work dating back to Manku and Motwani \cite{manku} --- see the survey \cite{streamingitemsetssurvey}. None of these works have been able to show
that these algorithms use less space than uniform random sampling of database rows, and our results justify why.
In particular, to our knowledge there has been no work establishing lower bounds on the space complexity of streaming algorithms for identifying 
approximate frequent itemsets that are better than the lower bounds that hold for the much simpler \emph{approximate frequent items} problem (a.k.a. the heavy hitters problem).
Note that the lower bounds that we establish in this work apply even to summaries computed by non-streaming algorithms.
\paragraph{Related Work by Price \cite{price}.} An earlier version of this manuscript by the first three authors  \cite{early} gave lower bounds on the size of sketches for frequent itemset mining that are quantitatively weaker than the bounds presented here \cite{early}.  In work subsequent to \cite{early}, and contemporaneous with the work presented here, Price \cite{price} discovered a short proof of an optimal lower bound for the \IFI\ sketching problem (defined in Section \ref{sectionDef} below) for itemsets $\itemset$ of size $|\itemset|=O(1)$. We compare our results to Price's in more detail in Section \ref{sec:overview}.

\eat{In particular, our lower bounds, which apply even to summaries computed by non-streaming algorithms, 
give (to our knowledge) the first proof that approximate frequent itemset computation (streaming or otherwise) requires strictly more resources than streaming algorithms for computing approximate frequent items.
}

 \eat{
Another relevant research direction involves sketching valuation functions \cite{Badanidiyuru2012SVF,BalcanH12},
as $g(T)  = 1-f_T$ is valid valuation function (that is, it is both both monotone and sub-additive).
Any additive $\eps$ approximation to $g$ yields a similar approximation for $f$. 
While according to \cite{Badanidiyuru2012SVF,BalcanH12} sketching general valuation functions efficiently is impossible up to an approximation factor of $O(n^{1/2})$
{\bf MM:  is this multiplicative?  We've gone last sentence from talking about additive to multiplicative, which seems odd.  What do we mean by impossible here -- impossible under some assumptions, or up to some small space, or generally?  Not clear},
itemset frequency functions could still obtain better approximation because they represent a restricted subclass of all valuation functions.
Indeed, subsampling already shows that better sketches for frequency functions are possible.
}

\subsection{Notation and Problem Statements} \label{sectionDef}
\label{sec3} 
Throughout, $\database \in \left(\bits^{\cols}\right)^{\rows}$
will denote a binary database consisting of $\rows$ rows and $\cols$ columns, or attributes.
We denote
the set $\{1, \dots, \cols\}$ by $[\cols]$. 
An \emph{itemset} $\itemset \subseteq [\cols]$ is a subset of the attributes; abusing notation,
we also use $\itemset$
to refer to the \emph{indicator vector} in $\{0, 1\}^{\cols}$ whose
$i$th entry is 1 if and only if $i \in \itemset$. 
We refer to any itemset $\itemset$ with $|\itemset| = k$ as a \emph{$k$-itemset}.
The $i$th row of $\database$ will be denoted by $\database(i)$, and
the $j$th entry of the $i$th row will be denoted by $\database(i, j)$.
We say that a row \emph{contains} an itemset $\itemset$ if the row has a 1 in all columns
in $\itemset$.
The \emph{frequency} $f_{\itemset}(\database)$  of $\itemset$ is the fraction of rows of $\database$ that contain $\itemset$. 
Alternatively, $f_{\itemset}(\database) = \frac{1}{n} \sum_{i=1}^{n} \mathbb{I}_{\{T \subseteq \database(i)\}}$.
We use the simplified notation $f_{\itemset}$ instead of $f_{\itemset}(\database)$ when the meaning is clear.
Note that we may view a row $\database(i)$ of $\database$ as a one-row database in its own right; hence, 
$f_{\itemset}(\database(i))$ equals $1$ if $\database(i)$ contains $\itemset$, and equals $0$ otherwise.
 
 We consider \emph{four} different sketching problems that each capture a natural notion of approximate 
 itemset frequency analysis. 
All four sketching problems permit randomized sketching algorithms, and require that the sketching algorithm succeeds with high probability. 
Success can be interpreted in two different ways: (1) with high probability, for all $k$-itemsets the answer is correct; or (2) for each $k$-itemset, with high probability the answer is correct (but it may be very unlikely that one can recover accurate estimates for all $k$-itemsets from the sketch simultaneously). These two different notions of success are often termed, respectively, ``for all'' and ``for each'' in the compressed sensing literature \cite{andoni}.
Our first two problem definitions (Definitions \ref{def:ifi} and \ref{def:ife}) correspond to the ``for all'' notion, while 
our latter two problem problem definitions (Definitions \ref{def:sifi} and \ref{def:sife}) correspond to the weaker ``for each'' notion.

\begin{definition}[\IFI\ sketches]
\label{def:epskapprox2} \label{def:ifi}
A \IFI\ sketch is a tuple $(\sumalg, \recalg)$. 
The first term $\sumalg$ is a randomized sketching algorithm.
It takes as input a database $\database \in \left(\bits^{\cols}\right)^{\rows}$, a precision $\eps$, an itemset size $k$, and a failure probability $\delta$.
It outputs a \emph{summary} $\sumalg(\database,k,\eps,\delta) \in \{0, 1\}^{s}$ where $s$ is the size of the sketch in bits.
The second term is a deterministic query procedure $\recalg: \{0, 1\}^{s} \times \{0, 1\}^{\cols} \rightarrow \{0,1\}$. 
It takes as input a summary $\sumalg$ and a $k$-itemset $\itemset$ and outputs a single bit indicating whether $\itemset$ is frequent in $\database$ or not. 
More precisely, for a triple of input parameters $(k, \eps, \delta)$, the following two conditions must hold with probability $1-\delta$ over the randomness of the sketching algorithm $\sumalg$, for every database $\database$:
\begin{equation} \label{eq:ifione}
\forall \; \mbox{$k$-itemsets} \; T \mbox{\;s.t.\;} f_T > \eps, \;\; \recalg(\sumalg(\database, k, \eps, \delta),T) = 1, \text{ and}
\end{equation}
\begin{equation}  \label{eq:ifitwo}
\forall \; \mbox{$k$-itemsets} \; T \mbox{\;s.t.\;} f_T < \eps/2, \;\; \recalg(\sumalg(\database, k, \eps, \delta),T) = 0.
\end{equation}
Note that if  $\eps/2 \leq f_T \leq \eps$ then either bit value can be returned.
\end{definition}
\begin{definition}[\IFE\ sketches]
\label{def:ife}
A \IFE\ sketch is a tuple $(\sumalg, \recalg)$.
Here $\sumalg$ is defined as above but $\recalg : \{0, 1\}^{s} \times \{0, 1\}^{\cols} \rightarrow [0,1]$ outputs an approximate frequency.
To be precise, the pair $(\sumalg, \recalg)$ is a valid \IFE\ sketch for a triple of input parameters $(k, \eps, \delta)$ if for every database $\database$:
\begin{equation} \label{eq:ifedef}
\Pr [\forall \; \mbox{$k$-itemsets} \; T, \;\; | \recalg(\sumalg(\database, k, \eps, \delta),T) - f_T | \le \eps] \ge 1-\delta.
\end{equation}
\end{definition}
\begin{definition}[\SIFI\ sketches]
 \label{def:sifi}
A \SIFI\ sketch is identical to a \IFI\ sketch, except that Equations \eqref{eq:ifione} and \eqref{eq:ifitwo} are replaced
with the requirement that for every database $\database$ and any (single) $k$-itemset $T$:
\[
 \text{If } f_T > \eps, \text{ then } \recalg(\sumalg(\database, k, \eps, \delta),T) = 1 \text{ with probability at least } 1-\delta, \text{ and }
\]
\[
 \!\!\!\!\!\!\!\!\text{If } f_T < \eps/2, \text{ then } \recalg(\sumalg(\database, k, \eps, \delta),T) = 0 \text{ with probability at least } 1-\delta.
\]
\end{definition}
\begin{definition}[\SIFE\ sketches]
 \label{def:sife}
A \SIFE\ sketch is identical to a \IFE\ sketch, except that Equation \eqref{eq:ifedef} is replaced
with the requirement that for every database $\database$ and any (single) $k$-itemset $T$:
$\Pr [| \recalg(\sumalg(\database, k, \eps, \delta),T) - f_T | \le \eps] \ge 1-\delta.$
\end{definition}
\begin{definition} 
The space complexity of a sketch, denoted by $|\sumalg(n, d, k, \eps, \delta)|$, is the maximum sketch size generated by $\sumalg$ for any database with $n$ rows and $d$ columns. 
That is, $|\sumalg(n, d, k, \eps, \delta)| = \max_{\database \in \left(\{0, 1\}^d\right)^{n}} |\sumalg(\database,k,\eps,\delta)|$. 
\end{definition}
For brevity, we typically omit the parameters $(n, d, k, \eps, \delta)$ when the meaning is clear, and simply write $|\sumalg|$ to denote the space complexity of a sketch.

\subsection{Techniques}
At a high level our lower bounds are proven via the standard approach of encoding arguments.  That is, to prove that $\Omega(s)$ bit sketches are necessary to solve one of the problems above, we construct a distribution over $s$-bit databases $\database \in (\bits^d)^{s/d}$ such that 1) the entropy of the distribution is $\Omega(s)$ and 2) any itemset frequency sketch can be used to reconstruct the database.  Thus, the sketch must have size $\Omega(s)$ bits in the worst case.

For the simplest case of $k = 1$ and $\eps = 1/3$, it is easy to show that $\Omega(d)$ bits are necessary to solve every version of the itemset sketching problem, since any non-trivial estimation of the $1$-itemset frequencies is sufficient to exactly encode a database consisting of a single row.  
In order to prove larger space lower bounds for larger values of $k$ and smaller values of $\eps$, we must show that if we are given either $k$-itemset queries for $k = \omega(1)$, or sketches with accuracy $\eps = o(1)$, then we can encode databases consisting of more than one row.

To do so, we combine some new arguments with information-theoretic techniques that were previously developed to solve problems in \emph{differential privacy}~\cite{difpriv}.  The problem of constructing differentially private sketches for frequent itemset queries has been studied intensely in recent years (see e.g.~\cite{difpriv3, difpriv4, de, difpriv5, difpriv1, difpriv2, buv}; in these works frequent itemset queries are called monotone conjunction queries and sometimes contingency tables).  It turns out that the techniques in these works provide exactly the information-theoretic tools that we need to devise our encoding arguments.  Although the connection between the two problems appears coincidental, we believe that 
information-theoretic tools from the privacy literature may find future applications outside of privacy.\footnote{Indeed, our use of techniques from the privacy literature may not be purely coincidental, as there is a formal (though quantitatively loose) connection between the problems of developing private and non-private sketches for itemset frequency queries. Suppose that there is a way to take any dataset $D \in (\bits^d)^n$ and create a sketch $\mathcal{S}$ of $s$ bits that encodes the answer to every itemset query $f_T(D)$ to within additive error $\pm \eps$.  Then we can obtain a differentially private sketch that encodes the answers to every query $f_T(D)$ to within additive error $\eps + O(s/n)$ as follows.  Output a sketch $S$ with probability proportional to $\exp(-n \cdot \max_{T} | f_T(D) - \mathcal{Q}(S, T) |)$.  Our claims that the output $S$ will be differentially private and that with high probability the chosen sketch will have additive error $\eps + O(s/n)$ can both be proven by an elementary analysis, or by using the fact that this algorithm is a special case of the exponential mechanism~\cite{mt} and thus standard results can be applied.  Now, suppose that we had a lower bound saying that any differentially private algorithm must incur error $t/n$ (lower bounds in differential privacy are often of the form $t / n$ for some $t$ independent of $n$).  Then we would immediately obtain a lower bound saying that any $\eps$-accurate sketch for itemsets requires $s = \Omega(t - \eps n)$ bits.  Thus, accuracy lower bounds in differential privacy \emph{immediately} imply space lower bounds for the associated sketching problem.  However, this generic connection rarely leads to bounds that are quantitatively tight.}  

We now sketch roughly how our encoding arguments work.  Suppose we have already proven a lower bound of $\Omega(s)$ bits for any sufficiently accurate $k$-itemset sketch via an encoding argument.  For this informal discussion, the precise definition of accuracy will not be important.  We then have a high-entropy distribution $\mathcal{D}$ on databases $D \in (\bits^d)^{s/d}$ such that a sufficiently accurate sketch for $k$-itemset frequency queries on $D$ must encode $D$.  For example, as we discussed above, we can trivially start with a lower bound of $d$ bits for any non-trivial approximation to the $1$-itemset queries, although sometimes we will need to start with stronger lower bounds.  We then obtain a lower bound of $\Omega(t \cdot s)$ for any sufficiently accurate sketch for $k'$-itemset frequency queries using the following amplification technique, inspired by the technique in~\cite{buv} for amplifying lower bounds in differential privacy.  Roughly, the technique allows us to construct a new distribution $\mathcal{D'}$ on databases $D' \in (\bits^{2d})^{ts/d}$ such that any sufficiently accurate sketch for $k'$-itemset frequency queries can be used to reconstruct an accurate $k$-itemset sketch on each of $t$ different draws from the distribution $\mathcal{D}$ over $(\bits^d)^{s/d}$.  Since we have started by assuming that such a sketch for $k$-itemset frequency queries requires $\Omega(s)$ bits, we conclude that any sufficiently accurate sketch of for $k'$-itemset frequency queries requires $\Omega(t \cdot s)$ bits of space.  The above outline is overly simplified, and the resulting reconstruction will only be approximate.  Thus, we need to make sure that the approximation is good enough for our arguments to go through, especially in some of our bounds that require applying the above amplification arguments twice, where the approximation becomes even weaker. 

\eat{
Unfortunately, there are two challenges with using the above observation in a generic way to prove lower bounds on sketch size.  First, this generic reduction incurs a quantitative loss in parameters, and does not lead to optimal space lower bounds.  Second, the appropriate error lower bounds for differentially private itemset sketches do not appear in the literature for many of the parameter regimes of interest here.  

Fortunately, most lower bound arguments that appear in the privacy literature have the following special form: 
they show that, given accurate answers to a particular class of queries, one can reconstruct essentially the entire database (and therefore any
mechanism giving accurate answers to this class of queries cannot be private). Such arguments are fundamentally information-theoretic in nature, 
and it is perhaps unsurprising that they can be used to prove optimal space lower bounds even in \emph{non-private} settings. 

Nonetheless, the arguments appearing in the extant literature do not directly yield space lower bounds for \IFI\ sketches, and do not 
yield tight space lower bounds for \IFE\ sketches for most of the parameter
regimes we are interested in. The technical contribution of our work is to extend and synthesize these techniques to obtain optimal or essentially optimal space lower bounds for both \IFI\ and \IFE\ sketches, in a much wider range of parameter regimes than what has been considered in prior work.
}

\section{Na\"ive upper bounds}
\label{naiveAlgorithms}
In the following we describe three trivial sketching algorithms.

\begin{definition}[\releasedb] 
This algorithm simply releases the database verbatim. In other words, the function $\sumalg$ is the identity and $\recalg$ is a standard database query.
\end{definition}
The space complexity of \releasedb\ is clearly $|\sumalg| = O(nd)$ and it produces exact estimates for both \IFE\ and \IFI\ sketches and their For-Each analogs.

\begin{definition}[\releaseanswers] 
This algorithm precomputes and stores the results to all possible queries.
\end{definition}
Since there are ${d \choose k}$ possible $k$-itemset queries, the space complexity of \releaseanswers\ is $|\sumalg| = O({d \choose k})$ for \IFI\ sketches 
and their For-Each analogs, and $|\sumalg| = O\left( {d \choose k}\log(1/\eps)\right )$ for  \IFE\ sketches and their For-Each analogs.
The extra $\log(1/\eps)$ factor is needed to represent frequencies as floating point numbers up to precision $\eps$.

\begin{definition}[\subsample] 
This algorithm samples rows uniformly at random with replacement from the database. The samples constitute the sketch $\sumalg(\database,k,\eps,\delta)$. 
The recovery algorithm $\recalg(\sumalg(\database), T)$ returns the frequency of $T$ in the sampled rows via a standard database query.
\end{definition}
\begin{lemma}[Subsampling] \label{lemma:subsample}
\subsample\ outputs
\begin{itemize}
\item a valid \IFI\ sketch of space complexity $|\sumalg| \!=\! O\!\left (\eps^{-1} \cols \log\! \left({{\cols \choose k}/\delta} \!\right)\!\right)$,
\item a valid \IFE\ sketch of space complexity $|\sumalg| = O\left ( \eps^{-2} \cols \log \left ({{\cols \choose k}/\delta} \right) \right)$,
\item a valid \SIFI\ sketch of space complexity $|\sumalg| = O\left ( \eps^{-1} \cols \log(1/\delta) \right)$, and 
\item a valid \SIFE\ sketch of space complexity $|\sumalg| = O\left ( \eps^{-2} \cols \log(1/\delta) \right)$.
\end{itemize}
\end{lemma}

\begin{proof}
Each row sample requires $d$ bits to represent. Thus to prove each of the above statements it suffices to bound the number of row samples required to ensure the accuracy goal.  We can do so using standard probabilistic inequalities:  Chernoff bounds for sums of independent random variables and union bounds. We will need the following standard forms of the Chernoff bound.

\begin{lemma} \label{chernoff1}
Suppose $X_1,\dots,X_s$ are independent random $\{0,1\}$-valued random variables with expectation $p$, and let $\overline{X} = \frac{1}{s}\sum_{i=1}^{s} X_i$. Then for any $\eps < 2e-1$, $\mathbb{P}[\overline{X} \not\in [(1-\eps)p, (1+\eps)p] \leq 2\exp(-sp\eps^2/4)$.
\end{lemma}

\begin{lemma}
\label{chernoff2}
Suppose $X_1,\dots,X_s$ are independent random $\{0,1\}$-valued random variables with expectation $p$, and let $\overline{X} = \frac{1}{s}\sum_{i=1}^{s} X_i$. Then for any $\eps < 1$, $\mathbb{P}[\overline{X} \not\in [p - \eps, p + \eps] \leq 2\exp(-2s\eps^2)$.
\end{lemma}

\smallskip
\noindent \emph{\SIFI\ sketches:}
Fix a dataset $\database$ and an itemset $T$ and let $p = f_{T}(\database)$.  Consider drawing $s$ independent uniform random samples of rows $\database'(1),\dots,\database'(s)$ with replacement from $\database$.  Let $\database'$ be the database consisting of the $s$ row samples. For $i=1,\ldots,s$, define the random variable $X_i = 1$ if $T \subseteq D'(i)$
and 0 otherwise. Let $\overline{X} = f_T(\database') = \frac{1}{s} \sum_{i=1}^{s} X_i$.  Since the samples $\database'(i)$ are independent, the random variables $X_i$ are independent.  Moreover, for every $i$, $\mathbb{E}[X_i] = p$, and by linearity of expectation $\mathbb{E}[\overline{X}] = p$. Then by Lemma \ref{chernoff1}, we have that
$$
\mathbb{P}[f_T(\database') \not\in [p/2, 2p]] \leq 2\exp(-sp/16).
$$
The right hand side will be at most $\delta$ if $s \geq 16 \ln(2/\delta)/p$ for a sufficiently large constant $C$.  From this bound, we can deduce that, for \SIFI\ sketches, it suffices to take $s = O(\eps^{-1}\log(1/\delta))$ row samples to ensure that the accuracy requirement of Definition \ref{def:sifi} is satisfied.

\smallskip
\noindent \emph{\SIFE\ skeches:}
The setup is the same, except that we apply Lemma \ref{chernoff2} instead of Lemma \ref{chernoff1} to obtain: 
$$
\mathbb{P}[f_T(\database') \not\in [p - \eps, p + \eps]] \leq 2\exp(-2s\eps^2).
$$
The right hand side will be at most $\delta$ if $s \geq \eps^{-2}\ln(2/\delta)$.  Thus, for \SIFE\ sketches it suffices to take $s = O(\eps^{-2}\log(1/\delta))$ row samples to ensure that the accuracy requirement of Definition \ref{def:sife} is satisfied.

\smallskip
\noindent \emph{\IFI\ sketches:}
By the analysis above, we know that for any $\delta'> 0$ and any itemset $T$, $f_{T}(\database')$ is accurate with probability at least $1-\delta'$ if $s = O(\eps^{-1}\log(1/\delta'))$.  Thus, by a union bound 
$$
\mathbb{P}[\exists \text{ } T \subseteq [d], |T| = k,\; \textrm{$f_{T}(\database')$ is not accurate}]
\leq \binom{d}{k} \mathbb{P}[\textrm{$f_{T}(\database')$ is not accurate}] \leq \binom{d}{k} \delta'.
$$
Now, by setting $\delta' = \delta / \binom{d}{k}$, we can see that it suffices to take $s = O(\eps^{-1} \log(\binom{d}{k}/\delta))$ samples to ensure accuracy.

\smallskip
\noindent \emph{\IFE\ sketches:}
Here we apply the same union bound argument to our analysis of \SIFE\ sketches.  We can easily see that it suffices to take $s = O(\eps^{-2} \log(\binom{d}{k}/\delta))$.
\end{proof}

For any setting of the parameters $(\rows, \cols, k, \eps, \delta)$, the minimal space usage among the above three trivial algorithms constitutes our na\"ive upper bound for all four
sketching problems that we consider,
formalized in Theorem \ref{thm:upperbound} below.

\begin{theorem} \label{thm:upperbound}
(a) For any $(k, \eps, \delta)$, there is a randomized algorithm that, given any database $\database \in  (\bits^{\cols})^{\rows}$, outputs a \IFI\ sketch of size
$O\left(\min\left\{\rows \cols, {\cols \choose k}, \eps^{-1} \cols \log \left({\cols \choose k}/\delta\right)\right\}\right).$

\medskip \noindent (b) For any $(k, \eps, \delta)$, there is a randomized algorithm that, given any database $\database \in (\bits^{\cols})^{\rows}$, outputs a \IFE\ sketch of size $O\left(\min\left\{\rows \cols, {\cols \choose k}  \log(1/\eps), \eps^{-2} \cols \log\left({\cols \choose k}/\delta\right)\right\}\right).$

\medskip \noindent (c) For any $(k, \eps, \delta)$, there is a randomized algorithm that, given any database $\database \in (\bits^{\cols})^{\rows}$, outputs a \SIFI\ sketch of size
$O\left(\min\left\{\rows \cols, {\cols \choose k}, \eps^{-1} \cols \log(1/\delta) \right\}\right).$

\medskip \noindent (d) For any $(k, \eps, \delta)$, there is a randomized algorithm that, given any database $\database \in (\bits^{\cols})^{\rows}$, outputs a \SIFE\ sketch of size
$O\left(\min\left\{\rows \cols, {\cols \choose k}  \log(1/\eps), \eps^{-2} \cols \log(1/\delta) \right\}\right).$

\end{theorem}

\section{Lower Bounds}
In this section, we turn to proving lower bounds on the size of \IFI\ and \IFE\ sketches. 
Notice that the algorithms \releaseanswers\ and \subsample\ produce sketches
whose size is independent of $n$; hence, it is impossible to prove lower bounds that
grow with $n$. Consequently, we state our lower bounds in terms of the parameters
$(d, k, 1/\eps)$, with all of our lower bounds holding as long as $n$ is sufficiently
large relative to these three parameters. This parameter regime --- with $n$ a sufficiently large polynomial in $d$, $k$, and $1/\eps$ --- is consistent
with typical usage scenarios, where the number of rows in a database far
exceeds the number of attributes. In our formal theorem statements, we 
make explicit precisely how large a polynomial $n$ must be in terms of $d$, $k$, and $1/\eps$
for the lower bound to hold.

Each of our lower bounds also requires $d$, $k$, and $1/\eps$ to satisfy
certain mild technical relationships with each other --- for example,
Theorems \ref{thm:easylowerboundRevised} and \ref{thm:easylowerboundReviseds} require that $1/\eps < {d/2 \choose k-1}$.
In all cases,
the assumed technical relationship between the parameters is necessary or close to necessary
for the claimed lower bound to hold. For instance, the $\Omega(d/\eps)$ lower bound of Theorems \ref{thm:easylowerboundRevised} and  \ref{thm:easylowerboundReviseds} is false
for $1/\eps \gg {d/2 \choose k-1}$, as the algorithm \releaseanswers\ would output
a sketch of size $o(d/\eps)$ in this parameter regime.

\subsection{Overview of the Lower Bounds}
\label{sec:overview}
We now formally state all of the lower bounds that we prove,
and place our results in context. Throughout this section, 
we assume that the failure probability $\delta$ of the sketching algorithm 
is a constant less than one.
%
\paragraph{Resolving the complexity of \SIFI\ sketches.}
The main result stated in this section is a relatively simple
$\Omega(d/\eps)$ lower bound for the \SIFI\ sketching problem (Theorem \ref{thm:easylowerboundReviseds} below), which is essentially optimal despite its simplicity.
For expository purposes, it will be convenient to first state an analogous lower bound for
the (harder) \IFI\ sketching problem. In Section \ref{firstproofs}, we prove the For-All lower bound first, and then explain how to modify the proof
to handle the For-Each case.

\begin{theorem}\label{thm:easylowerboundRevised}
Let $k \geq 2$. 
Suppose that $1/\eps \le {\cols/2 \choose k-1}$, and the failure probability $\delta < 1$ is constant.
Then for $n \geq 1/\eps$, the space complexity of any valid \IFI\ sketch is $|\sumalg(n, k, d, \eps, \delta)| = \Omega(d/\eps)$. 
\end{theorem}

\begin{theorem}\label{thm:easylowerboundReviseds}
Let $k \geq 2$. 
Suppose that $1/\eps \le {\cols/2 \choose k-1}$, and $\delta < 1/3$.
Then for $n \geq 1/\eps$, the space complexity of any valid \SIFI\ sketch is $|\sumalg(n, k, d, \eps, \delta)| = \Omega(d/\eps)$. 
\end{theorem}

Theorem \ref{thm:easylowerboundReviseds} is tight whenever it applies (i.e., when $1/\eps < {d/2 \choose k-1}$), as it matches the $O(d/\eps)$
upper bound obtained by the algorithm \subsample\ for the \SIFI\ sketching problem (see Theorem \ref{thm:upperbound}). And
the algorithm \releaseanswers\ achieves a summary size of ${d \choose k}$ for the \SIFI\ sketching problem, which is
asymptotically optimal when $1/\eps \geq {d/2 \choose k-1}$ and $k=O(1)$. Therefore,
Theorems \ref{thm:upperbound} and \ref{thm:easylowerboundReviseds} together precisely resolve the complexity of \SIFI\ sketches for all values of $d$ and $\eps$, when $k=O(1)$.
\paragraph{Resolving the complexity of \IFI\ sketches.}
Theorem \ref{thm:easylowerboundRevised} is tight for \IFI\ sketches when $1/\epsilon$ is small relative
to the other input parameters $n$ or $d$. In particular, when $n=1/\eps$, 
 \releasedb\ provides a trivial matching sketch that is $O(n\cols) = O({\cols/\epsilon})$ bits in size.
In addition, when $k=O(1)$ and $1/\eps \geq {\cols/2 \choose k-1}$,
$\releaseanswers$ provides a matching sketch that is $O({d \choose k}) = O(d/\epsilon)$ bits in size.
 The tightness of Theorem \ref{thm:easylowerboundRevised} in these parameter regimes is arguably surprising, as it shows that the \IFI\ sketching problem is
 \emph{equivalent} in complexity to its For-Each analog in these regimes.
 
However, when $1/\eps \ll {d/2 \choose k-1}$, Theorem \ref{thm:easylowerboundRevised} is not tight for \IFI\ sketches,
because it has suboptimal dependence on $d$ and $k$. Our main result for the \IFI\ sketching problem establishes a tight lower bound, matching
the $O(\eps^{-1} d \log{d \choose k})$ upper bound 
for the problem obtained by the algorithm \subsample.



\begin{theorem} \label{thm:hardlowerboundformal1}
Let $k \geq 3$, and suppose that $1/\eps = O\left({d/3 \choose \lfloor(k-1)/2\rfloor}\right)$ and the failure probability $\delta < 1$ 
is a constant.
Then for any $n \geq k d \log(d/k)/\eps$, the space complexity of any valid \IFI\ sketch is $|\sumalg(n, k, d, \eps, \delta)| = \Omega(k d \log(d/k)/\eps)$.  
\end{theorem}

\paragraph{Comparison to Price's work.}
Price independently proved an $\Omega(d \log(d)/\eps)$ lower bound on the size of \IFI\ sketches for $k \geq 2$ and $1/\eps \leq d^{.99}$ \cite{price}. 
This matches the lower bound of Theorem \ref{thm:hardlowerboundformal1} for any $k=O(1)$, but not for $k=\omega(1)$. 
One advantage of Price's result is that it holds for $k=2$ and any value of $\eps$ satisfying $1/\eps \leq d^{.99}$;
in contrast, Theorem \ref{thm:hardlowerboundformal1} holds for $k \geq 3$. We remark that our proof of Theorem \ref{thm:hardlowerboundformal1} 
actually establishes the  $\Omega(d \log(d))$ lower bound even for $k=2$, but 
our extension of the proof to sub-constant values of $\eps$ requires $k \geq 3$.


\paragraph{Essentially resolving the complexity of \IFE\ sketches.}
Theorem \ref{thm:hardlowerboundformal2} below establishes a lower bound for the \IFE\ sketching problem with a quadratically
stronger dependence on $1/\eps$, relative to the linear dependence that 
is necessary and sufficient for the \IFI\ problem. 
Our lower bound matches the space usage of \subsample\ (cf. Lemma \ref{lemma:subsample}) up to a factor of $\log_{(q)}(1/\eps)$ for any desired constant $q > 0$, where $\log_{(q)}$ denotes
the logarithm function iterated $q$ times (e.g.~$\log_{(3)}(x) = \log \log \log (x)$).\footnote{\label{woodrufffootnote}An anonymous reviewer has pointed out that 
subsequent work of Van Gucht et al. \cite[Theorem 6]{woodruff}
can be combined with our amplification techniques to remove the $\log_{(q)}(1/\eps)$ factor in the bound of Theorem \ref{thm:hardlowerboundformal2}, when $\eps > 1/\sqrt{d}$;
the resulting bound matches the size of the sketch
produced by \subsample\ up to a constant factor. (Note that
Theorem \ref{thm:hardlowerboundformal2} holds even for values of $\eps$ smaller than $1/\sqrt{d}$, whenever $k > 3$.)
In more detail, the proof of \cite[Theorem 6]{woodruff} 
implies an $\Omega(d/\eps^2)$ lower bound on the size of any \IFE\ sketch answering 99\% of all $2$-itemset queries, when $\eps>1/\sqrt{d}$. 
By combining this result with our amplification techniques, it is possible to prove an $\Omega\left(\frac{k d \log(d/k)}{\eps^2}\right)$ lower bound
on the size of any \IFE\ sketch answering all $k$-itemset queries for $k\geq 3$, whenever $\eps > 1/\sqrt{d}$.}

\begin{theorem} \label{thm:hardlowerboundformal2}
Fix any integer constants $c \geq 2$, $q \geq 1$, and let $\delta < 1$ be a constant.  Let $k \geq c+1$, let $d, \eps$ satisfy $1/\eps^2 \leq d^{c-1}/\log_{(q)}(1/\eps^2)$, and let $\delta < 1$ be a constant. Let $\sumalg$ be a \IFE\ sketching algorithm capable of answering all $k$-itemset frequency queries to error $\pm \eps$ on databases $\database \in \left(\{0, 1\}^d\right)^n$ for any $n > d \log(d/k) \log_{(q)}(1/\eps^2) / \eps^2$. Then $|\sumalg(n, d, k, \eps, \delta)| = \Omega\left(\frac{k d \log(d/k)}{\eps^2 \log_{(q)}(1/\eps)}\right).$
\end{theorem}

For illustration, consider fixing the constants $c = q = 10$.  Then theorem says that for every $k \geq 11$, if $1/\eps = O(d^{4.5}/\log_{(10)}(d))$, and the database size $n$ is sufficiently large, then there is a space lower bound of \\
$|\mathcal{S}| = \Omega\left(kd\log(d/k)/\left(\eps^2 \log_{(10)}(1/\eps)\right)\right)$.

Finally, we use a simple reduction to show that the above theorem also implies a
lower bound for the For-Each case, which is optimal up to a factor of $\log_{(q)}(1/\eps)$.\footnote{Footnote \ref{woodrufffootnote}
implies that one can remove the $\log_{(q)}(1/\eps)$ factor in the bound of Theorem \ref{thm:hardlowerboundformal3} when $\eps > 1/\sqrt{d}$.
The resulting bound matches the size of the sketch
produced by \subsample\ in this parameter regime up to a constant factor.}

\begin{theorem} \label{thm:hardlowerboundformal3}
Fix any integer constants $c \geq 2$, $q \geq 1$.  Let $k \geq \max\{3, c+1\}$, let $d, \eps$ satisfy $1/\eps^2 \leq d^{c-1}/\log_{(q)}(1/\eps^2)$, and let $\delta < 1/2$ be a constant. Let $\sumalg$ be an \SIFE\ sketching algorithm capable of answering any (single) $k$-itemset frequency queries to error $\pm \eps$ on databases $\database \in \left(\{0, 1\}^d\right)^n$ for any $n > d \log(d/v) \log_{(q)}(1/\eps^2) / \eps^2$. Then $|\sumalg(n, d, k, \eps, \delta)| = \Omega\left(\frac{d}{\eps^2 \log_{(q)}(1/\eps)}\right).$
\end{theorem}

\subsection{Lower Bound Proofs for Itemset-Frequency-Indicator Sketches}
\label{firstproofs}
\subsubsection{First Lower Bound Proofs: Theorems \ref{thm:easylowerboundRevised} and \ref{thm:easylowerboundReviseds}}
\label{sec:ifilbs}
We begin by proving our two simplest bounds (Theorems \ref{thm:easylowerboundRevised} and \ref{thm:easylowerboundReviseds}). Recall that the former
applies to \IFI\ sketches, and the latter applies even to their For-Each analogs. 
The proofs consider databases in which even a single appearance of an itemset already makes it frequent.  We show that, unsurprisingly, essentially no compression is possible in this setting. (For simplicity,
we assume that $1/\eps$ is an integer throughout.)
\begin{proof}[Proof of Theorem \ref{thm:easylowerboundRevised}]
Our proof uses an encoding argument. 
Consider the following family of databases. 
There will be $1/\eps$ possible settings for each row; as $n \geq
1/\eps$, some rows may be duplicated.  For expository purposes, we begin by describing the setting 
with $n = 1/\eps$, in which case there are no duplicated rows.
The first $d/2$ columns in each row contain a unique set of exactly $k-1$ attributes.
The last $d/2$ attributes in each row are unconstrained. 
The only minor technicality is that to ensure that each row can receive a unique set of $k-1$ items from the first $d/2$ attributes, we 
require $1/\eps \le {d/2 \choose k-1}$. 

Given a valid \IFI\ or \IFE\ sketch for this database, one can recover all of the values $\database(i,j)$ where $j\ge d/2$ as follows.
For any $j \geq d/2$, let $T_{i,j}$ be a set of $k$ attributes, where the first $k-1$ attributes in $T_{i,j}$ correspond to the $k-1$ attributes
in the first $d/2$ columns in the $i$th row, and 
the final attribute in $T_{i, j}$ is $j$.
Notice that $T_{i,j} \in \database$ if and only if $\database(i,j) = 1$. Moreover, since $n = 1/\eps$  we have that $f_T \ge \eps$ if
and only if $\database(i,j) = 1$.
Given a valid \IFI\ or \IFE\ sketch for this database, one can iterate over all itemsets $T_{i,j}$ to recover all the values $\database(i,j)$ where $j\ge d/2$.
Since these are an unconstrained set of $d/(2\eps)$ bits, the space complexity of storing them (with $1-\Omega(1)$ failure probability) is $\Omega(d/\eps)$
by standard information theory.

For $n$ a multiple of $1/\epsilon$, we construct a database with
$1/\epsilon$ rows as above, and duplicate each row $n\epsilon$ times;
in this case we have $f_T \ge \eps$ if and only if $\database(i,j) =
n\epsilon$.  More generally, when $n \geq 1/\epsilon$, duplicating each row
at least $\lfloor n\epsilon \rfloor$ times, we have 
$f_T \ge \eps$ if and only if $\database(i,j) \geq
\lfloor n\epsilon \rfloor$.
\end{proof} 
We remark that the condition $1/\eps \le {\cols/2 \choose k-1}$ can be relaxed to $1/\eps \le {\alpha \cols \choose k-1}$ for any constant $\alpha < 1$,
by a simple modification of the proof.
We now extend the argument used to prove Theorem \ref{thm:easylowerboundRevised} to the For-Each case.

\begin{proof}[Proof of Theorem \ref{thm:easylowerboundReviseds}]
Recall that in the setting of one-way randomized communication complexity, there are two parties, Alice and Bob.
Alice has an input $x \in \mathcal{X}$, Bob has an input $y \in \mathcal{Y}$, and Alice and Bob 
both have access to a public random string $r$. Their goal is to compute $f(x, y)$
for some agreed upon function $f \colon \mathcal{X} \times \mathcal{Y} \rightarrow \{0, 1\}$. 
Alice sends a single message $m(x, r)$ to Bob.
Based on this message, Bob outputs a bit, which is required to equal $f(x, y)$ with probability at least $2/3$. 

We consider the well-known \textsc{INDEX} function. In this setting, Alice's input $x$ is an $N$-dimensional binary vector, Bob's input $y$ is an index in $[N]$, and 
$f(x, y)=x_y$, the $y$'th bit of $x$. It is well-known that one-way randomized communication protocols for \textsc{INDEX} require communication $\Omega(N)$ \cite{ablayev}. 
We show how to use any \SIFI\ sketching algorithm $\sumalg$ to obtain a one-way
communication protocol for \textsc{INDEX} on vectors of length $N=(d/2) \cdot 1/\eps$, with communication proportional to $|\sumalg(n, d, \eps, k, \delta)|$. 

Specifically, let $(n, d, k, \eps, \delta)$ be as in the statement of the theorem. Consider any Boolean vector $x \in \{0, 1\}^{N}$, where $N=(d/2) \cdot 1/\eps$. 
We associate each index $y \in [N]$ with a unique $k$-itemset $T_y \subseteq [d]$ of the following form: 
the first $k-1$ attributes in $T_{y}$ are each in $[d/2]$, and 
the final attribute in $T_{y}$ is in $\{d/2+1, \dots, d\}$.
The proof of Theorem \ref{thm:easylowerboundRevised}
established the following fact: given any vector $x \in \{0, 1\}^{N}$, there exists a database $\database_x$ with $d$ columns and $n$ rows 
satisfying the following two properties for all $y \in [N]$: \begin{equation}
 x_y = 1 \Longrightarrow f_{T_y}(\database_x) \geq \eps, \text{ and } x_y = 0 \Longrightarrow f_{T_y}(\database_x) = 0 < \eps/2. \label{eq:inproof} \end{equation}

Hence, we obtain a one-way randomized protocol for the \textsc{INDEX} function as follows: Alice sends to Bob $\sumalg(\database_x, k, \eps, \delta)$ at 
a total communication of $|\sumalg(n, d, \eps, k, \delta)|$ bits,
and Bob outputs $\recalg(\sumalg(\database_x, k, \eps, \delta), T_y)$. It follows immediately from Equation \eqref{eq:inproof} and Definition \ref{def:sifi} that 
Bob's output equals $x_y$ with probability $1-\delta$.
We conclude that $|\sumalg(n, d, \eps, k, \delta)| = \Omega(N) = \Omega(d/\eps)$, completing the proof.
\end{proof}

\subsubsection{A Tight Lower Bound for \IFI\ Sketches: Proof of Theorem \ref{thm:hardlowerboundformal1}}
\label{sec:tightifi}

Our proof of Theorem \ref{thm:hardlowerboundformal1} is inspired by an approach from~\cite{buv} for ``bootstrapping'' two weak privacy lower bounds into a stronger lower bound.

\begin{proof}[Proof of Theorem \ref{thm:hardlowerboundformal1}]
We begin by proving an $\Omega(k\cols \log(d/k))$ lower bound for $\eps = 1/50$ and every $k \geq 2$ (the specific choice $\eps=1/50$ is for convenience; the lower bound holds for any suitably small constant). We then
use the ideas underlying the proof of Theorem \ref{thm:easylowerboundReviseds} to extend the lower bound to sub-constant
values of $\eps$, for any $k\geq 3$. 
\paragraph{The case of $\eps=1/50$.} We begin by recalling a basic combinatorial fact about $k$-itemset frequency queries on databases with $d$ attributes. 

\begin{fact}\label{fact}
For any $k' \geq 1$, let $v= k'\cdot \log(d/k')$. There exist strings $x_1, \dots, x_v \in \{0,1\}^d$ such that for every string $s\in \{0,1\}^v$, there is a $k'$-itemset $\itemset_{s}$ such that $f_{\itemset_s}(x_i) = s_i$ for all $i \in [v]$.
\end{fact}
\begin{proof}[Proof of Fact \ref{fact}]
It is well-known that the set of $k'$-itemset frequency queries (equivalently, $k'$-way monotone conjunction queries), when evaluated on $d$-bit vectors (i.e., on $d$-attribute database rows), has VC dimension at least $k' \cdot \log(d/k')$. 
The desired strings $x_1, \dots, x_v$ are simply the shattered set whose existence is guaranteed by having VC dimension $v$. It is also not difficult to directly construct the shattered set; we provide such a direct construction in Appendix \ref{app:vc} for completeness.
\end{proof}

Let $k' = k-1$, let $v = (k-1)\log(d/(k-1))$, and let $x_1,\dots,x_v \in \bits^{\cols}$ be the strings promised by Fact~\ref{fact}.  Let $y_1, \dots, y_v \in \{0, 1\}^d$ be an arbitrary set of $v$ strings of length $d$.  We will show how to construct a database $\database$ with $v$ rows and $2d$ columns such that, with probability at least $1-\delta$, at least 96\% of the $d v$ bits in $(y_1, \dots, y_v)$ can be reconstructed from any \IFI\ sketch $\sumalg(\database, k, 1/50, \delta)$.  This will imply an $\Omega(dv)$ space lower bound.
Specifically, define row $i$ of $\database$ to be \begin{equation} \label{eq:ddef} \database(i) := (x_i, y_i).\end{equation} That is, the first $d$ bits in $\database(i)$ are equal to $x_i$ and the last $d$ bits are equal to $y_i$. 

The key observation behind the reconstruction of the $y_i$'s is that, given \emph{exact} answers to all $k$-way itemset frequency queries, one can compute the inner product between the last $d$ columns of $\database$ and any desired vector. Moreover, it is easy to see that, given sufficiently many inner products, any column of $\database$ can be exactly reconstructed. However, \IFI\ sketches do not provide exact answers to itemset frequency queries; they merely indicate whether the frequency of an itemset is larger than $\eps$ or smaller than $\eps/2$. Nonetheless, we show that in order to reconstruct 96\% of the bits in any given column of $\database$, it is enough to know, for sufficiently many vectors, whether the inner product of the column with the vector is larger than $\eps=1/50$ or smaller than $\eps/2=1/100$. 

In more detail, fix a column $j \in \{d+1, \dots, 2d\}$.  Let $t := (y_{1, j}, \dots, y_{v, j}) \in \{0,1\}^v$ be the bits in this column of $\database$. Fix a string $s$ in $\{0,1\}^v$, and for any $j \in [d]$, let $\itemset_{s, j} = \itemset_s \cup \{j\}$, where $\itemset_s$ is defined as in Fact \ref{fact}.  We claim that the correct answer to the itemset frequency query $f_{\itemset_{s, j}}(\database)$ is $\langle s, t \rangle/v$.  To see this, notice that by the definition of $\database$ (Equation \eqref{eq:ddef}), $\itemset_{s, j}$ is contained in any row $i$ of $\database$ such that $s_i = y_{ij} = 1$, and $\itemset_{s, j}$ is not contained in any other rows. Hence, $f_{\itemset_{s, j}}(\database) = \frac{1}{v} |\{i : s_i=1, y_{ij}=1\}| =\langle s, t \rangle/v$.  

It follows that any \IFI\ sketch $\cal{S}$ of $\database$ that provides answers with error parameter $\eps=1/50$ for all $k$-itemsets, provides a bit $b_{s}$ for every $s \in \{0, 1\}^v$ such that the following holds: $b_s=1$ if $\langle s, t \rangle/v > \eps$, and  $b_s=0$ if $\langle s, t \rangle/v < \eps/2$. The following lemma implies that we can use the $b_s$ values to reconstruct a vector 
$t'$ that is close to $t$ in Hamming distance. 

\begin{lemma} \label{sublemma}
Suppose for every $s \in \{0,1\}^{v}$, we are given a bit $b_s$ satisfying $b_s=1$ if $\langle s, t \rangle/v > \eps$, and  $b_s=0$ if $\langle s, t \rangle/v < \eps/2$. Let $t' \in \{0,1\}^v$ be any vector that is consistent with all of the $b_{s}$ values, in the sense that
$\langle s, t' \rangle/v > \eps$ for all $s$ such that $b_s = 1$, and $\langle s, t' \rangle/v < \eps/2$ for all $s$ such that $b_s=0$. Then the Hamming distance between $t$ and $t'$ is
at most $v/25$. 
\end{lemma}
\begin{proof}[Proof of Lemma \ref{sublemma}]
Consider any vector $t' \in \{0, 1\}^v$ such that $t$ and $t'$ differ in more than $v/25$ fraction of bits. Then there must a set of coordinates $S \subseteq [v]$ of size at least $v/50$ such that at least one of the two conditions is satisfied: (a) $t'_j= 1$ and $t_j =0$ for all $j \in S$, or (b) $t'_j=0$ and $t_j=1$ for all $j \in S$. 

Assume without loss of generality that Condition (a) is satisfied (the proof in the case that Condition (b) is satisfied is analogous). Consider the vector $s \in \{0, 1\}^v$ that is the indicator vector of $S$. Then $\langle s, t\rangle = 0$, so $b_s=0$. However, $\langle s, t' \rangle \geq v/50$. This implies that $t'$ is not consistent with the value $b_s$ returned by the \IFI\ sketch, proving the lemma.
\end{proof}

Lemma \ref{sublemma} implies that, for any $k \geq 2$, given the \IFI\ sketch $\mathcal{S}(\database, k, 1/50, \delta)$, we can recover at least 96\% of the bits of $(y_1, \dots, y_v)$ with probability at least $1-\delta$. Suppose we let $(y_1, \dots, y_v)$ be the error-corrected encoding 
of a vector $(y'_1, \dots, y'_z) \in \{0, 1\}^z$, using a code with constant rate that is uniquely decodable from $4\%$ errors (e.g. using a Justesen code \cite{justesen}). Then $z=\Omega(v)$, and it follows from the above that $(y'_1, \dots, y'_z)$ can be \emph{exactly} reconstructed from  $\mathcal{S}(\database, k, 1/50, \delta)$ with probability at least $1-\delta$. 
Hence, $\cal{S}(\database)$ allows for exact reconstruction of $z=\Omega(d v)$ arbitrary bits with probability $1-\delta$. Basic information theory then implies that
$|\mathcal{S}| = \Omega(d v) = \Omega\left(k d  \log(d/k)\right)$. 
\paragraph{The case of $\eps=o(1)$.}
For any $k \geq 3$, suppose that we are given a \IFI\ sketching algorithm $\sumalg$ that is capable of answering \IFI\ queries with error parameter $\eps$ for all $k$-itemsets. For simplicity, we assume $k$ is odd. At a high level, we show that, given $m=\frac{1}{50\eps}$ independent databases $\database_i$, each with $v$ rows and $2d$ columns, we can construct a single ``larger'' database $\database$ with $mv$ rows and $3d$ columns such that the following holds: for \emph{every} $\database_i$, $\sumalg(\database, k, \eps, \delta)$ can be used to answer all \IFI\ queries on $\database_i$ with error parameter $\eps'=1/50$ for all $(k+1)/2$-itemsets. Since we have assumed $k\geq 3$, it holds that $(k+1)/2 \geq 2$; hence, we can apply our earlier analysis to conclude that any such summary for $\database_i$ contains $\Omega\left(kd\log(d/k)\right)$ bits of information, in the sense that it can encode an arbitrary bit vector of this length. It follows that $\sumalg(\database, k, \eps, \delta)$ contains $\Omega\left(kd\log(d/k)/\eps\right)$ bits of information, proving the theorem. Details follow.

Let $\itemset_1, \dots, \itemset_m \subseteq [d]$ be distinct $((k-1)/2)$-itemsets (note that as in the proof of Theorem \ref{thm:easylowerboundRevised}, we can indeed choose $m$ such $\itemset_i$'s as long as $1/\eps < {d \choose (k-1)/2}$). 
Now consider any $m$ independent databases $\database_1, \dots, \database_m$, each with $v$ rows and $2d$ columns. We construct a new $(mv) \times 3d$ database $\database$ by appending the $d$-bit indicator vector of $\itemset_i$ to each row of $\database_i$, and letting $\database$
be the concatenation of all of the resulting databases.

For each $\itemset_i$, let $\itemset' \subseteq [3d]$ be defined via $\itemset'_i = \{j+2d: j \in \itemset_i\}$. That is, $\itemset'_i$
is simply $\itemset$ ``shifted'' to operate on the final $d$ of the $3d$ attributes over which the ``larger'' database $\database$ is defined.
Let $\itemset^* \subseteq [2d]$ be any $(k+1)/2$-itemset, and for each $i \in [m]$, let $\itemset^*_i \subseteq [3d]$ be the $k$-itemset defined via: $\itemset^*_i = \itemset^* \cup \itemset'_i$. Observe that $f_{\itemset^*}(\database_i) = m \cdot f_{\itemset^*_i}(\database')$. 
Hence, $f_{\itemset^*_i}(\database) > \eps$ if and only if $f_{\itemset^*}(\database_i) > 1/50$, and 
$f_{\itemset^*_i}(\database) < \eps/2$ if and only if $f_{\itemset^*}(\database_i) < 1/100$.
That is, one can use $\sumalg(\database, k, \eps, \delta)$ to answer all \IFI\ queries on $\database_i$ with error parameter $1/50$ for all $k$-itemsets (this holds simultaneously for all $i$ with probability $1-\delta$).

By the argument for the case $\eps=1/50$, this implies that $\sumalg(\database, k, \eps, \delta)$ can be used to losslessly encode an arbitrary
vector of length $\Omega(d \cdot v \cdot m ) = \Omega( \left(k d  \log(d/k)/\eps\right)$, and thus $|\sumalg(\database, k, \eps, \delta)| = \Omega\left( \left(k d  \log(d/k)/\eps\right)\right)$.
 This completes the proof of Theorem \ref{thm:hardlowerboundformal1}.
\end{proof}

\section{Lower Bounds Proofs for Itemset-Frequency-Estimator Sketches}
\subsection{The For-All Case: Proof of Theorem \ref{thm:hardlowerboundformal2}} 

\subsubsection{Informal Overview of the Proof}
For constant $k' \geq 2$, an $\tilde{\Omega}(d/\eps^2)$ lower bound on the size of \IFE\ sketches follows fairly directly 
from existing work in the literature on differential privacy (cf. Kasiviswanathan et al. \cite{difpriv4}; we refer to this work as KRSU). 
The idea of KRSU's result is the following. 
Itemset frequency queries are a linear class of queries, in the sense that we can represent any database as a vector $z$ (in which each entry of $z$ corresponds to a possible record in $\bits^d$ and its value is the number of such records in the database), and the vector of answers to all $k'$-itemset
frequency queries on $z$ can be written as $A z$ for some matrix $A$. Given a vector $y$ of approximate answers to these queries, on can try to reconstruct $z$ via the approximation $\hat{z} = A^{-1} y$, where $A^{-1}$ denotes the Moore-Penrose pseudo-inverse of $A$ (this is essentially reconstruction via $L_2$-distance minimization). If the matrix $A$ has a ``nice'' spectrum, then it is possible to bound the distance between $\hat{z}$ and $z$. If this distance is small enough, then any description of $y$ contains many bits of information, since it essentially encodes an entire database $z$. 

However, KRSU do not actually look at the matrix $A$ corresponding to $k'$-itemset frequency queries. Instead, they look at a matrix $M^{(k')}$ they define as follows. Consider a database $\database$ with $n$ rows and $k'$ columns, where the first $k'-1$ columns of $\database$ are generated at random. For any fixed setting of the first $k'-1$ columns, the vector of answers to $k'$-itemset frequency queries on $\database$ are a function only of column $k'$ of $\database$. Denoting column $k'$ of $\database$ by $x$, these answers can be written in the form $M^{(k')} x$, for a particular matrix $M^{(k')}$ derived from the first $k'-1$ columns of $\database$. 

KRSU show that $M^{(k')}$ behaves a lot like a matrix with truly random entries from $\{0, 1\}$, and hence has a ``nice'' spectrum (with high probability over the random choice of the first $k'-1$ columns of $\database$). This ensures that $\hat{x} = A^{-1} y$ is a ``good'' approximation to the last column $x$ of $\database$ as long as all answers in $y$ have error $\eps \lesssim \sqrt{n}$. Put another way, if the error in the answers is $\eps$, then it is possible to reconstruct column $k'$ of $\database$ as long as the number of rows is at most (roughly) $1/\eps^2$. 

This shows that one can use a summary providing $\eps$-approximate answers to
all $k'$-itemset frequency queries on a database with $k'$ columns and $1/\eps^2$ rows to reconstruct $\tilde{\Omega}(1/\eps^2)$ arbitrary bits.
It is possible to extend this argument to databases with $d+k'-1$ columns and $\tilde{\Omega}(1/\eps^2)$ rows,
yielding a $\tilde{\Omega}(d/\eps^2)$ lower bound on the size of \IFE\ sketches for such databases. 

Our main contribution for the \IFE\ problem is to combine such an $\tilde{\Omega}(d/\eps^2)$ lower bound for sketches for $k'$-way marginals with a technique for ``amplifying'' the lower bound to $\tilde{\Omega}(k \cdot \log(d/k) \cdot d/\eps^2)$ for $(k+k')$-way marginal queries (we used essentially the same amplification technique, which was inspired by work of Bun et al. \cite{buv} in the context of differential privacy, in Section \ref{sec:tightifi}). This technique says that, given $v=k \cdot \log(d/k)$ databases $\database'_1, ..., \database'_v$, each with $d+k'-1$ columns and $1/\eps^2$ rows, we can construct a bigger database $\database$ such that one can use $\eps$-approximate answers to all $(k+k')$-way marginal queries on $\database$ to obtain $\eps$-approximate answers to the $k'$-way marginals on every database $\database'_i$. 

This amplification technique actually requires the $\tilde{\Omega}(d/\eps^2)$ lower bound for $k'$-way marginals to hold even if the answer vector $y$ only has error $\eps$ ``on average'', rather than having error at most $\eps$ for every single answer. Hence, we cannot directly use the KRSU lower bound in our argument. In fact, to reconstruct a database from answers that have error at most $\eps$ only ``on average'', one cannot use $L_2$ distance minimization as in KRSU's lower bound argument, since $L_2$-minimization is highly sensitive to a few answers having large error. Fortunately, De \cite{de} shows how to use $L_1$-minimization to establish an $\tilde{\Omega}(1/\eps^2)$ lower bound even in the setting in which answers are only required to have error at most $\eps$ ``on average''. We use his techniques to obtain a lower bound suitable for our argument.

\subsubsection{Proof Details}
In the context of differential privacy, De~\cite{de}, building on~\cite{difpriv4,rudelson}, described an algorithm
for reconstructing a database $\database$, given sufficiently accurate answers to all $k$-itemset frequency queries on $\database$. 
In our terminology, De's result establishes that 
any \IFE-sketch can be used to losslessly encode $\Omega(\frac{d}{\eps^2\log_{(q)}(1/\eps)})$ bits of information.  Here $\log_{(q)}(\cdot)$ denotes the logarithm function iterated $q$ times.
Formally, we use the following slight refinement of De's result.  
\begin{lemma}[Variant of Theorem 5.12 of \cite{de}] \label{lemma:de} For any constant integers $k \geq 2$ and $q \geq 1$, there exists a constant $\gamma = \gamma(k, q) > 0$ and a distribution $\mu$ over $k$-itemset queries such that the following holds. 

Let $d$ and $\eps$ be parameters satisfying $1/\eps^2 \leq d^{k-1}/\log_{(q)}(1/\eps^2)$.
Suppose $\sumalg$ is any summary algorithm that can answer a $1-\gamma$ fraction of all $k$-itemset frequency 
queries under $\mu$ on databases with $d$ columns to error $\pm \eps$.  Then there exists a $b = b(d, \eps) = \Omega(d/\eps^2 \log_{(q)}(1/\eps))$,  an $n=n(d, \eps) = O(\log_{(q)}(d)/\eps^2)$, a database-generation algorithm $\alg$ that takes as input a Boolean vector $y \in \{0, 1\}^{b}$ and outputs a database $\alg(y) \in \left(\{0, 1\}^{d}\right)^n$, and a decoding algorithm $\alg'$ such that $\alg'$ outputs $y$ with high probability given $\sumalg(\alg(y))$.
\end{lemma}
We prove Lemma \ref{lemma:de} in Appendix \ref{app:de}.  

Lemma \ref{lemma:de} alone is enough to yield a lower bound of $\Omega\left(d/\eps^2 \log_{(q)}(1/\eps)\right)$ on the 
size of \IFE\ sketches capable of answering all $k$-itemset frequency queries to error $\pm \eps$, for any $k \geq 2$. 
The technical contribution of this section is to ``bootstrap'' this result to obtain a lower bound of $\Omega\left(k d \log(d/k)/\eps^2\log_{(q)}(1/\eps)\right)$
bits, which improves over the bound that follows from a direct application of Lemma \ref{lemma:de} even for $k=3$. This lower bound is essentially optimal, matching the $O(k d \log(d/k)/\eps^2)$ upper bound achieved
by algorithm \subsample\ up to a $\log_{(q)}(1/\eps)$ factor (for an arbitrarily large constant $q$). 




\begin{proof}[Proof of Theorem \ref{thm:hardlowerboundformal2}]
Let $v = (k-c)\log(d/(k-c))$ and $x_1, \dots, x_v \in \{0,1\}^d$ be the strings promised by Fact \ref{fact} applied with $k'=k-c$. 
Recall that for every vector $s \in \{0, 1\}^v$, there is a $(k-c)$-itemset $\itemset_s$ such that $f_{\itemset_s}(x_i)=s_i$ for
all $i \in [v]$.  Recall that $c \geq 2$ is a parameter of the theorem.

Let $\gamma = \gamma(c, q)$ be the constant in Lemma~\ref{lemma:de}.  Fix $\eps' = 100\eps / \gamma$.  Suppose that we are given $v$ strings $y_1, \dots, y_v \in \{0, 1\}^b$, where $b=b(d, \eps')$ is as in Lemma \ref{lemma:de} for $k = c$. Let $ \alg(y_i) = \database_i \in \left(\{0, 1\}^d\right)^n$, where $\alg$ is the database generation algorithm
promised by Lemma \ref{lemma:de}. We show how to construct a single ``large'' database $\database$ with $2d$ columns
and $n v$ rows such that 96\% of the bits of $y_1, \dots, y_v$ can be recovered from $\sumalg(\database, k, \eps, \delta)$.  Note that Lemma~\ref{lemma:de} applies, since $c \geq 2$.  Also note that $\gamma$ is indeed a constant since we required that $c$ and $q$ are constants; therefore $\eps' = O(\eps)$.

\medskip
\noindent \textbf{Definition of $\database$.} Recall that $\database_i(j)$ denotes the $j$th row of $\database_i$. Define $\database'_i$ to be the database with $2d$ columns and $n$ rows defined via $\database'_i(j) = (x_i, \database_i(j))$.
That is, $\database'_i$ is obtained from $\database_i$ by appending the string $x_i$ to the front of every row. We define
$\database$ to be the concatenation of all of the $\database'_i$ databases. We index the $v n$ rows of $\database$
as $(i, j): i \in [v], j \in [n]$. 

\medskip
\noindent \textbf{Reconstructing $y_1, \dots, y_v$ from $\sumalg(\database)$.}  For any $c$-itemset query $\itemset \subseteq [d]$, let $z_{\itemset}$ denote the vector \linebreak $z_{\itemset}=(f_{\itemset}(\database_1), \dots, f_{\itemset}(\database_v))$. Let $s \in \{0, 1\}^v$ be any vector. Define $\itemset' = \itemset'(\itemset, s) \subseteq [2d]$ to be the $k$-itemset
whose indicator vector is the concatenation of the indicator vectors of $\itemset_s$ and $\itemset$; that is, 
$\itemset' := \itemset_s \cup \{j + d: j \in  \itemset\}$. 

We claim that $\frac1v \langle s, z_{\itemset} \rangle = f_{\itemset'(s, \itemset)}(\database)$. To see this, note that
$\itemset'$ is contained in row $(i, j)$ of $\database$ if and only if $s_i=1$ and $\itemset$ is contained in row $j$ of $\database_i$.
Hence, 
\begin{eqnarray}
\frac{1}{v} \langle s, z_{\itemset} \rangle  & = & \frac{1}{v} \cdot \sum_{i \in [v]: s_i=1} f_\itemset(\database_i)\\
& = & \frac{1}{v} \sum_{i \in [v]} \frac{1}{n} |\{j \in [n]:  \itemset \text{ is contained in } D_{i}(j)\}|\\
& = & \frac{1}{nv} \sum_{(i, j) \in [v] \times [n]} |\{\itemset' \text{ is contained in row } (i, j) \text{ of } \database \}|\\
& = &  f_{\itemset'(\itemset, s)}(\database).
\end{eqnarray}

Hence, from any \IFE\ sketch $\sumalg(\database, k, \eps, \delta)$, one can compute for every $c$-itemset $\itemset$, an estimate $\hat{f}_{\itemset'(\itemset, s)}$ 
satisfying $|\hat{f}_{\itemset'(\itemset, s)} - \frac1v \langle s, z_{\itemset}\rangle| \leq \eps$.  The following lemma
describes why these estimates are useful in reconstructing $y_1, \dots, y_v$.

\begin{lemma} \label{lemma:end}
Fix a $c$-itemset $\itemset \subset [d]$.
Given values $\hat{f}_{\itemset'(s, \itemset)}$
satisfying $|\hat{f}_{\itemset'(s, \itemset)} - \frac1v \langle s, z_{\itemset}\rangle| \leq \eps$ for all $s \in \{0, 1\}^v$, it is possible
to identify a vector $\hat{z}_{\itemset} \in [0, 1]^v$ satisfying
$\frac1v \|\hat{z}_{\itemset} - z_{\itemset}\|_1 \leq 4 \eps$.
\end{lemma}
\begin{proof}
Consider the algorithm that outputs any vector $\hat{z}_{\itemset} \in [0, 1]^v$ satisfying the following property:

\begin{equation} \label{eq:thefinish} \text{For all } s \in \{0, 1\}^v,  \left| \frac1v \langle \hat{z}_{\itemset} , s \rangle - \hat{f}_{\itemset'(s, \itemset)} \right| \leq \eps. \end{equation}

Note that at least one such vector always exists, because setting $\hat{z}_{\itemset} = z_{\itemset}$ satisfies Equation \eqref{eq:thefinish}. Thus, the algorithm always produces some output.

We claim that any $\hat{z}_{\itemset}$ output by the algorithm satisfies $\frac1v \|\hat{z}_{\itemset} - z_{\itemset}\|_1 \leq 4 \eps$.
Indeed, suppose otherwise. Define $s^{(1)} \in \{0, 1\}^v$ via $s^{(1)}_i = 1$ if and only if $\hat{z}_{\itemset, i} \geq z_{\itemset, i}$ and
$s^{(2)} \in \{0, 1\}^v$ via $s^{(2)}_i = 1$ if and only if $z_{\itemset, i} > \hat{z}_{\itemset, i}$. 
Then either $\frac{1}{v} \left( \langle \hat{z}_{\itemset}, s^{(1)} \rangle - \langle z_{\itemset}, s^{(1)} \rangle \right) > 2\eps$, or
$\frac{1}{v} \left( \langle \hat{z}_{\itemset}, s^{(2)} \rangle - \langle z_{\itemset}, s^{(2)} \right) > 2 \eps$. Assume without loss of generality 
that the former case holds.  Then 

$$ \left| \frac1v \langle \hat{z}_{\itemset} , s^{(1)} \rangle - \hat{f}_{\itemset'(s^{(1)}, \itemset)} \right| \geq 
\frac{1}{v} \left( \langle \hat{z}_{\itemset}, s^{(1)} \rangle - \langle z_{\itemset}, s^{(1)} \rangle \right) - \left| \frac1v \langle z_{\itemset} , s \rangle - \hat{f}_{\itemset'(s, \itemset)}\right|  > 2\eps - \eps = \eps,$$
where the inequality holds by the triangle inequality. But this contradicts the assumption that $\hat{z}_{\itemset}$ satisfies~\eqref{eq:thefinish}.
\end{proof}

For each $c$-itemset $\itemset$, let $\hat{z}_{\itemset}$ be as in Lemma \ref{lemma:end}. We think of $\hat{z}_{\itemset, i}$ as an estimate of $z_{\itemset, i} = f_{\itemset}(\database_i)$. Lemma \ref{lemma:end} guarantees that \emph{for any distribution $\mu$ over $c$-itemsets $\itemset$}, this estimate has error at most $4\eps$ on average, when the averaging is done over a random $c$-itemset $\itemset$ chosen according to $\mu$, and databases $\database_i$. In symbols:
\[\mathbf{E}_{\itemset \leftarrow \mu} \mathbf{E}_{i \in v} |\hat{z}_{\itemset, i} - z_{\itemset, i}| = \sum_{\itemset} \mu(\itemset)  \left(\mathbf{E}_{i \in v} |\hat{z}_{\itemset, i} - z_{\itemset, i}|\right) \leq \sum_{\itemset} \mu(\itemset) \cdot 4\eps \leq 4\eps.\]
Here, $\mathbf{E}_{\itemset \leftarrow \mu}$ denotes the expectation operation when $\itemset$
is chosen according to the distribution $\mu$, and the penultimate inequality holds by Lemma \ref{lemma:end}. 

By Markov's inequality we conclude that 
for at least 96\% of the databases $\database_i$, the estimates $\hat{z}_{\itemset, i}$ have error
at most $100\eps$ on average, where the averaging is over the choice of $\itemset$ according
to distribution $\mu$.  That is,
for at least $96\%$ of databases $\database_i$, it holds that 
\begin{equation} \label{eq:mercy} \mathbf{E}_{\itemset \leftarrow \mu}  \mathbf{E}_{i \in v} |\hat{z}_{\itemset, i} - z_{\itemset, i}| \leq 100\eps.\end{equation}

For any $\database_i$ satisfying Equation \eqref{eq:mercy} and any $\gamma > 0$, another application of Markov's inequality implies that the $|\hat{z}_{\itemset, i} - f_{\itemset}(\database_i)| \leq 100\eps/\gamma = \eps'$ for a $1-\gamma$ fraction of all $c$-itemsets $\itemset$ under distribution $\mu$. By Lemma \ref{lemma:de}, this implies that $y_i$ can be exactly recovered from the $\hat{z}_{\itemset, i}$ values, using algorithm $\alg'$. 

Since $96\%$ of the $y_i$ vectors can be \emph{exactly} recovered, it follows that
at least $96\%$ of the $b v$ total bits in the vectors of $y_1, \dots, y_v$ can be recovered. 
Suppose we let the $bv$ bits in the collection of vectors $(y_1, \dots, y_v)$ be the error-corrected encoding 
of a single vector $(y'_1, \dots, y'_z) \in \{0, 1\}^z$, using a code with constant rate that is uniquely decodable from $4\%$ errors (e.g. using a Justesen code \cite{justesen}). Then $z=\Omega(bv)$, and it follows from the above that $(y'_1, \dots, y'_z)$ can be \emph{exactly} reconstructed from
$\sumalg(\database, k, \eps, \delta)$ with probability $1-\delta$.
Basic information theory then implies that $$|\sumalg(\database, k, \eps, \delta)| = \Omega(bv) = \Omega\left(\frac{k d \log(d/k)}{(\eps')^2 \log_{(q)}(1/(\eps'))}\right) = \Omega\left(\frac{k d \log(d/k)}{\eps^2 \log_{(q)}(1/\eps)}\right),$$ where we have used the fact that $\eps' = O(\eps).$  This completes the proof of the theorem.
\end{proof}

\subsection{The For-Each Case: Proof of Theorem \ref{thm:hardlowerboundformal3}}
Recall that Theorem \ref{thm:hardlowerboundformal3} establishes   
a lower bound against \SIFE\ sketches that is tight up to a $\log_{(q)}(d)$ factor. We prove Theorem \ref{thm:hardlowerboundformal3}
via a simple argument that shows how to transform any \SIFE\ sketch into a \IFE\ sketch 
with a modest increase in space. This allows us to transform Theorem \ref{thm:hardlowerboundformal2} 
into the claimed lower bound against \SIFE\ sketches. 
\begin{proof}[Proof of Theorem \ref{thm:hardlowerboundformal3}]
Suppose that we are given an \SIFE\ sketching algorithm $\sumalg$ using space $|\sumalg|$. We show how to transform
$\sumalg$ into a \IFE\ sketching algorithm $\sumalg'$ using space $O\left(|\sumalg| \cdot \log{d \choose k}\right) = O\left(|\sumalg(k, \eps, \delta)| \cdot k \cdot \log(d/k)\right)$. It then follows from Theorem \ref{thm:hardlowerboundformal2} that  $|\sumalg| = \Omega\left(\frac{d}{\eps^2 \log_{(q)}(1/\eps)}\right)$. 

The \IFE\ sketching algorithm $\sumalg'$ simply outputs $10 \cdot \log\left({d \choose k}/\delta\right)$ independent copies of $\sumalg(\database)$ 
(i.e., using fresh randomness for each of the $10 \cdot \log\left({d \choose k}/\delta\right)$ runs of $\sumalg$).
Given any $k$-itemset $\itemset$, the query procedure $\recalg'$ associated with $\sumalg'$ simply runs the 
query procedure $\recalg$ associated with $\sumalg$ on each of the copies of $\sumalg(\database)$, and outputs the median
of the results. Since each copy of $\sumalg$ outputs an estimated frequency that has error at most $\eps$
with probability $1-\delta > 1/2$, standard Chernoff Bounds imply that for any fixed $k$-itemset $\itemset$,
the median estimate will have error at most $\eps$ with probability at least $1-\delta/{d \choose k}$. 
A union bound implies that the median estimate will have 
error at most $\eps$ for all ${d \choose k}$ itemsets with probability at least $1-\delta$. Thus,
$\sumalg'$ is a \IFE\ sketching algorithm with failure probability at most $\delta$.
\end{proof}


\eat{

given ``accurate'' answers to all $k$-itemset frequency queries on a database $\database$,
it is possible to reconstruct a large constant fraction of the entries of $\database$. 

To explain this result, we must introduce some notation

\begin{lemma}
For a parameter $\eps$, let $\database$ be a database with $n=1/\eps^2$ rows and $d_1$ columns chosen uniformly at random from $\left(\{0, 1\}^{d_1}\right)^{n}$. 
Let $y \in \{0, 1\}^n$ be any vector, and let $\database'$ be the database obtained by adding an extra column $y$ to $\database$, so that $\database'$ has $n$ rows and $d$ columns.

For any positive constants $q$ and $k_1$, there exists a constant $\gamma = \gamma(q, k_1) > 0$ such that the following holds.

Let $d, \eps$ be parameters such that $1/\eps^2 \leq d^{k_1-1} \cdot \log_{(q)}(1/\eps^2)$. 
There is a distribution $\mu$ over databases $\database \in \left(\{0, 1\}^d\right)^{1/\eps^2}$ and an algorithm $\alg$ with the following property.
With probability $1-\exp(-\Omega(d))$ over the random choice of $\database$, the following holds.
Let $y \in \{0, 1\}^n$ be any vector, and let $\database'$ be the database obtained by adding an extra column $y$ to $\database$, so that $\database'$ has $n$ rows and $d$ columns. 

Given a list of values $\hat{f}_{\itemset}

\end{lemma}

Our proof builds on a fundamental result of Kasiviswanathan et al. \cite{difpriv4} that establishes the following. For a parameter $\eps$, let $\database$ be a database with $n=1/\eps$ rows and $d-1$ columns chosen uniformly at random from $\left(\{0, 1\}^{d-1}\right)^{n}$. 
Let $y \in \{0, 1\}^n$ be any vector, and let $\database'$ be the database obtained by adding an extra column $y$ to $\database$, so that $\database'$ has $n$ rows and $d$ columns. 
Kasiviswanathan et al. \cite{difpriv4} described an algorithm $\alg$ with the following input-output behavior.

\medskip \noindent \textbf{Input to $\alg$}: $\database$, as well as a list of values $\{\hat{f}_\itemset: |\itemset|=k, \itemset \subseteq [d]$ satisfying $|\hat{f}_{\itemset} - f_{\itemset}(\database')| \leq 1/\sqrt{n} = 1/\eps$. 

\medskip \noindent \textbf{Output of $\alg$:} A vector $y' \in \{0, 1\}^n$.

\begin{lemma}[\cite{difpriv4}]
With high probability over the choice of $\database$, the following holds. For any vector $y \in \{0, 1\}^{n}$, 
$\alg$ will output a vector $y'$ of fractional Hamming distance $1-o(1)$ from $y$. 
\end{lemma} 

}

\section{Conclusion}
We introduced four closely related notions that capture the problem of approximating itemset frequencies in databases. For all four problems, we studied the minimal size of sketches that permit a user to recover sufficiently accurate information about itemset frequencies. After identifying three 
naive algorithms that apply to all four problems, we turned to proving sketch size lower bounds. Our results establish that random sampling achieves optimal or essentially optimal sketch size for all four problems. This stands in contrast to several seemingly similar problems, such as identifying
approximate frequent items in data streams, and various matrix approximation problems, for which uniform sampling is not an optimal sketching algorithm. 

 We proved our lower bounds by adapting and extending techniques developed in the literature on differentially privacy
 data analysis. It is an interesting open question whether there are other problems in non-private data analysis that can 
 be resolved using techniques from the literature on differential privacy. 
 
In addition, our lower bound arguments specify a ``hard'' distribution over databases, for which it is impossible to improve upon the space usage of the uniform sampling sketching algorithm for answering approximate itemset frequency queries. But real-world databases are likely to be substantially more structured than the databases appearing in our hard distribution, and real-world query loads are likely to be highly non-uniform. In these settings, \emph{importance sampling} is a natural candidate for improving upon the space usage of the uniform sampling sketching algorithm. It would be interesting to identify rigorous yet realistic conditions on databases and query loads that allow for such an improvement.
Subsequent work by Lang et al. \cite{lang} takes some initial steps in this direction.

\medskip
\noindent \textbf{Acknowledgements.} The authors are grateful to Amit Chakrabarti, Graham Cormode, Nikhil Srivastava, and Suresh Venkatasubramanian for several helpful conversations during the early stages of this work.

\bibliographystyle{alpha}
\bibliography{itemsetrefs}

\appendix

\section{Proof of Fact \ref{fact}}
\label{app:vc}
We provide a direct construction
of the vectors whose existence is guaranteed by Fact \ref{fact}, restated here for convenience. 

\medskip \noindent \textbf{Fact \ref{fact}.} \emph{
For any $k' \geq 1$, let $v= k'\cdot \log(d/k')$. There exist strings $x_1, \dots, x_v \in \{0,1\}^d$ such that for every string $s\in \{0,1\}^v$, there is a $k'$-itemset $\itemset_{s}$ such that $f_{\itemset_s}(x_i) = s_i$ for all $i \in [v]$.}

\begin{proof}[Proof of Fact \ref{fact}]
For expository purposes, we first describe a set of vectors $w_1, \dots, w_{k'} \in \{0,1\}^{k'}$ that are ``shattered'' by $k'$-itemset frequency queries, i.e.,
for every string $s\in \{0,1\}^{k'}$, there is a $k'$-itemset $\itemset_{s}$ such that $f_{\itemset_s}(y_i) = s_i$ for all $i \in [k']$.
We then describe a set of vectors $y_1, \dots, y_{\log d} \in \{0,1\}^d$ that are shattered even by $1$-itemset frequency queries. Finally, we explain
how to ``glue together'' the $w_i$'s and $y_i$'s to obtain the full set $x_1, \dots, x_v \in \{0,1\}^d$ whose existence is claimed in
the statement of Fact \ref{fact}.

\medskip \noindent \textbf{Description of the} $w_i$\textbf{'s}. For each $i \in [k']$, define $w_i \in \{0, 1\}^{k'}$ via:

\[
\begin{cases}
w_{i, j}=1: & 1 \leq j \leq k', j \neq i\\
w_{i,j}=0 & j = i\\
\end{cases}
\]

To restate the above in matrix notation, we define the $k' \times k'$ matrix $W^{(k')}$ whose rows are the $w_i$'s via:

\[W^{(k')} := \begin{pmatrix} w_1 \\  w_2 \\  \vdots \\  w_{k'-1} \\  w_{k'} \end{pmatrix} = \begin{pmatrix} 0 & 1 & 1 & \hdots & 1 & 1 \\ 1 & 0 & 1 & \hdots & 1 & 1 \\ \vdots & \vdots & \vdots & \hdots & \vdots & \vdots   \\ 1 & 1 & 1 & \hdots & 0 & 1  \\ 1 & 1 & 1 & \hdots & 1 &0 \end{pmatrix} \]

For any string $s \in \{0, 1\}^{k'}$, let $\itemset_s := \{i: s_i=0\}$. It is straightforward to check that $f_{\itemset_s}(w_i)=s_i$ as desired.

\medskip \noindent \textbf{Description of the} $y_i$\textbf{'s}. For each $i \in [\log d]$, we define each $y_i \in \{0 , 1\}^d$ to ensure that
the matrix whose rows are the $y_i$'s contains every possible $\log(d)$-bit string as a column. In matrix notation, we define the $\log(d) \times d$ matrix
$Y^{(d)}$ via:

\[Y^{(d)} := \begin{pmatrix} y_1 \\  y_2 \\  \vdots \\  y_{\log(d)-1} \\  y_{\log d} \end{pmatrix} = \begin{pmatrix} 0 & 0 & 0 & \hdots & 1 & 1 \\ 0 & 0 & 0 & \hdots & 1 & 1\\ \vdots & \vdots & \vdots & \hdots & \vdots & \vdots  \\ 0 & 0 & 1 & \hdots & 1 & 1 \\ 0 & 1 & 0 & \hdots & 0 & 1\end{pmatrix} \]

For any string $s \in \{0, 1\}^{\log d}$, we interpret $s$ as the binary representation of an integer $\text{int}(s) \in \{0, \dots, d-1\}$, and define $\itemset_s := \{\text{int}(s)\}$. It is straightforward to check that $f_{\itemset_s}(y_i)=s_i$ as desired. 

\medskip \noindent \textbf{Description of the} $x_i$\textbf{'s}. 
Recall that $v=k' \cdot \log(d/k')$. Consider the $v \times d$ matrix $X$ whose rows are the $x_i$'s. We view this matrix
as a collection of sub-matrices, where each sub-matrix has dimension $k' \times (d/k')$. 
More specifically, let $\bfJ$ denote the $k' \times (d/k')$ matrix of all-ones. 
We define $X$ to be the matrix obtained from $W^{(d/k')}$ by replacing each entry of $W^{(d/k')}$ that is equal to $1$
with the matrix $\bfJ$, and replacing each entry of $W^{(d/k')}$ that is equal to 0 with the matrix $Y^{(d/k')}$. 
In more detail,
define:

\[X :=  \begin{pmatrix} x_1 \\  x_2 \\  \vdots \\  x_{k'-1} \\  x_{k'} \end{pmatrix} = \begin{pmatrix} Y^{(d/k')} & \bfJ & \bfJ & \hdots & \bfJ & \bfJ \\ \bfJ & Y^{(d/k')} & \bfJ & \hdots & \bfJ & \bfJ\\ \vdots & \vdots & \vdots & \hdots & \vdots & \vdots  \\ \bfJ & \bfJ & \bfJ & \hdots & Y^{(d/k')} & \bfJ \\ \bfJ & \bfJ & \bfJ & \hdots & \bfJ & Y^{(d/k')}\end{pmatrix} \]

 Given any vector $s \in \{0, 1\}^v$, we interpret $s$ as specifying $k'$ integers $\ell_1, \dots, \ell_{k'} \in \{0, \dots, d/k'\}$ in the natural way. 
 We view $[d]$ as the cross-product $[k'] \cdot [d/k']$, and associate each $j \in [d]$ with a pair $(r_1, r_2) \in [k'] \times [d/k']$ in the natural way.
 We then define $\itemset_s := \{ (i, \ell_i): i \in [k'] \}$. It is then straightforward to observe that
 $f_{\itemset_s}(x_i)=s_i$.
\end{proof}
\section{Proof of Lemma \ref{lemma:de}}
\label{app:de}

We restate Lemma \ref{lemma:de} for convenience, before providing its proof.

\medskip
\noindent \textbf{Lemma \ref{lemma:de}.} (Refinement of Theorem 5.12 of \cite{de}) \emph{For any constant integers $k \geq 2$ and $q \geq 1$, there exists a constant $\gamma = \gamma(k, q) > 0$ and a distribution $\mu$ over $k$-itemset queries such that the following holds.}

\emph{Let $d$ and $\eps$ be parameters satisfying $1/\eps^2 \leq d^{k-1}/\log_{(q)}(1/\eps^2)$.
Suppose $\sumalg$ is any summary algorithm that can answer a $1-\gamma$ fraction of all $k$-itemset frequency 
queries under $\mu$ on databases with $d$ columns to error $\pm \eps$.  Then there exists a $b = b(d, \eps) = \Omega(d/\eps^2 \log_{(q)}(1/\eps))$,  an $n=n(d, \eps) = O(\log_{(q)}(d)/\eps^2)$, a database-generation algorithm $\alg$ that takes as input a Boolean vector $y \in \{0, 1\}^{b}$ and outputs a database $\alg(y) \in \left(\{0, 1\}^{d}\right)^n$, and a decoding algorithm $\alg'$ such that $\alg'$ outputs $y$ with high probability given $\sumalg(\alg(y))$.}

As Lemma \ref{lemma:de} is a refinement of Theorem 5.12 of De's work \cite{de}, the presentation of our proof borrows heavily from De's.

\begin{proof}

We begin by defining the Hadamard product of matrices.

\begin{definition}[Hadamard product of matrices] Let $A_1, \dots, A_s \in \mathbb{R}^{\ell_i \times n}$.
Then, the Hadamard product of $A_1,  \dots , A_s$ is denoted by
$A = A_1 \circ A_2 \circ \dots \circ A_s \in \mathbb{R}^{L \times n}$, where $L =  \ell_1 \cdot \dots \cdot \ell_s$ and is defined as follows: 
Every row of $A$ is
identified with a unique element of $[\ell_1] \times \dots \times [\ell_s]$. For $i = (i_1, \dots i_s)$, define 
\[A[i, h] = \prod_{j=1}^s A_j[i_j, h]\]
where $A[i, h]$ represents the element in row $i$ and column $h$ of $A$.
\end{definition}

We will also require the definition of Euclidean sections, which play an important role
in the analysis of LP decoding algorithms. 

\begin{definition}[Euclidean Sections]
$V \subseteq \mathbb{R}^z$
is said to be a $(\delta, d', z)$ Euclidean Section if $V$ is a linear subspace of
dimension $d'$ and for every $x \in V$, the following holds:
\[\sqrt{z} \|x\|_2 \geq \|x\|_1 \geq \delta \sqrt{z} \|x\|_2.\]
A linear operator $A \colon \mathbb{R}^{d'} \rightarrow \mathbb{R}^{z}$
is said to be $\delta$-Euclidean if the range of $A$ is a Euclidean $(\delta, d', z)$
section.
\end{definition}

The following lemma follows directly from the proof of \cite[Lemma 5.9]{de}.

\begin{lemma}[Reformulation of Lemma 5.9 of \cite{de}] \label{biglemma} 
Let $d_0^{k-1} > n$. Suppose there exist 
Boolean matrices $A_1, \dots, A_{\ell-1} \in \mathbb{R}^{d_0 \times n}$ 
such that $A = A_1 \circ A_2 \circ \dots \circ A_{k-1}$, all the singular values of $A$ are at least $\sigma$,
and the range of $A$ is a $(\delta, n, d_0^{k-1}$)-Euclidean section. 

Let $\database_0$ denote the database with $n$ rows and $(k-1) \cdot d_0$ columns obtained
from the $A_i$'s as follows: the $j$th row of $\database_0$ is the concatenation of the $j$th row of each of the matrices
$A^T_1$, $A^T_2$, $\dots$, $A^T_{k-1}$. Let $\alg_1$ denote the database generation algorithm
that takes as input a Boolean vector $y \in \{0, 1\}^{n}$, and outputs the database $\database_1(y)$ with $n$ rows and $d_1:=(k-1)\cdot d_0 + 1$ 
columns obtained from $\database_0$ by appending an additional column equal to $y$. 

Then, there exists a constant 
$\gamma_1 = \gamma_1(\delta) > 0$, a distribution $\mu_1$ over $k$-itemsets $\itemset \subseteq [d_1]$, and a reconstruction algorithm $\mathcal{R}$ satisfying the following. Fix any $\zeta_1 \in o(\sqrt{n} \sigma /\sqrt{d_1^{k-1}})$. Suppose
$\mathcal{R}$ is given $\database_0$ and 
 approximate itemset frequencies $\hat{f}_{\itemset}$ for all $k$-itemsets $\itemset \subseteq [d_1]$.
 Let $S_1$ denote the set of all $k$-itemsets $\itemset$ satisfying
$n\cdot|\hat{f}_{\itemset} - f_{\itemset}(\database_2(y))| \leq \zeta_1$, and suppose that $\sum_{\itemset \in S_1} \mu(\itemset) \geq 1-\gamma_1$.
Then $\mathcal{R}$ outputs a vector $\hat{y}$ of Hamming distance $o(n)$ from $y$. 
 \end{lemma}

We use Lemma \ref{biglemma} to establish the following stronger statement. 

\begin{lemma}\label{biglemma2} Let $A_1, \dots, A_{k-1}$, $\delta$, $\database_0$, $\mu_1$, $d_0$, $d_1$, $n$, $\gamma_1$, and $\zeta_1$ be as in Lemma \ref{biglemma}.
There is a $b=b(d_0, n) \in \Omega(d_0 \cdot n)$ and database generation algorithm $\alg_2$ that takes as input a Boolean vector $y' \in \{0, 1\}^{b}$,
and outputs a database $\database_2$ with $n$ rows and $d_2:=(k-1)\cdot d_0 + d_0 = k \cdot d_0$ 
columns such that the following holds.

There exists a constant 
$\gamma_2 = \gamma_2(\delta) > 0$,  a distribution $\mu_2$ over $k$-itemsets $\itemset \subseteq [d_2]$, and a reconstruction algorithm $\mathcal{R}_2$ satisfying the following. Suppose
$\mathcal{R}_2$ is given $\database_0$ and 
 approximate itemset frequencies $\hat{f}_{\itemset}$ for all $k$-itemsets $\itemset \subseteq [d_2]$. 
 Let $S_2$ denote the set of all $k$-itemsets $\itemset$ satisfying 
$n\cdot|\hat{f}_{\itemset} - f_{\itemset}(\database_2(y'))|  \leq \zeta_1$, and suppose that $\sum_{\itemset \in S_2} \mu_2(\itemset) \geq 1-\gamma_2$.
 Then $\mathcal{R}_2$ outputs $y'$. 
 \end{lemma}
 \begin{proof}
Let $\alg_2$ be the database generation algorithm
that takes as input a Boolean vector $y' \in \{0, 1\}^{b}$, and first replaces $y'$ with an error-corrected encoding $\text{Enc}(y') \in \{0, 1\}^{d_0 \cdot n}$ of
$y'$, using an error-correcting code of constant rate that is uniquely decodable from $2\%$ errors. 
$\alg_2$ then outputs the database $\database_2(y')$ with $n$ rows and $d_2$ 
columns obtained from $\database_0$ by appending $d_0$ additional columns, with the first additional column equal to the first $n$ bits of $\text{Enc}(y')$,
the second additional column equal to the second $n$ bits of $\text{Enc}(y')$, and so on. We refer to the $d_0$ attributes corresponding to these additional columns as \emph{special attributes}.
Similarly, we call an itemset $\itemset \subseteq [d_2]$ \emph{special} if $\itemset$ contains exactly one special attribute.
 
  For each $i \in [d_0]$, let $y^{(i)} \in \{0, 1\}^{n}$ denote the vector $(\text{Enc}(y')_{(i-1) \cdot n+1}, \dots, \text{Enc}(y')_{i \cdot n})$;
  that is, $y^{(i)}$ is the $i$th ``block'' of $n$ bits from $\text{Enc}(y')$. 
 Let  $\database_1(y^{(i)})$ be as in the statement of Lemma \ref{biglemma}. Note that
 $\database_1(y^{(i)})$ is a sub-database of $\database_2(y')$, in the sense that $\database_1(y^{(i)})$
 equals $\database_2(y')$ with several columns removed. Hence, for any 
 $k$-itemset $\itemset_1 \subseteq [d_1]$, there is a unique itemset $g_i(\itemset_1) \subseteq [d_2]$
 such that $f_{\itemset_1}(\database_1(y^{(i)})) = f_{g_i(\itemset_1)}(\database_2(y'))$. Notice
 that $g_i(\itemset_1)$ is a special itemset, for any $i$ and $\itemset_1$. Moreover, the $g_i$'s are all invertible: 
 for any special itemset $\itemset_2 \subseteq [d_2]$,
 there is a unique itemset $h(\itemset_2) \subseteq [d_1]$ and a unique $i$ satisfying $g_i(h(\itemset_2))=\itemset_2$. 
   
 Let $\gamma_2 = \gamma/100$. We define the distribution $\mu_2$ over $k$-itemsets $\itemset \subseteq [d_2]$ as follows. 
\[\mu_2(\itemset) = 
\begin{cases} 0 & \text{ if } \itemset \text{ is not special.} \\
(1/d_0) \cdot \mu_1(h(\itemset)) & \text{ if } \itemset \text{ is special.}
\end{cases} \]
  
As per the hypothesis of the lemma, suppose
$\mathcal{R}_2$ is given $\database_0$ and 
 approximate itemset frequencies $\hat{f}_{\itemset}$ for all $k$-itemsets $\itemset \subseteq [d_2]$. 
 Let $S_2$ denote the set of all $k$-itemsets $\itemset$ satisfying 
$n \cdot |\hat{f}_{\itemset} - f_{\itemset}(\database_2(y'))|  \leq \zeta_1$, and suppose that $\sum_{\itemset \in S_2} \mu_2(\itemset) \geq 1-\gamma_2$.
 
 The recovery algorithm $\mathcal{R}_2$ will reconstruct $y'$ by first constructing
 a vector $y'' \in \{0, 1\}^{d_0 \cdot n}$ such that the fractional Hamming distance between $y''$ and $\Enc(y')$ is at most .02,
 and then running the decoding algorithm for the error-correcting code on $y''$.
 $\mathcal{R}_2$ constructs the vector $y''$ as follows. For each $i \in [d_0]$,
 $\mathcal{R}_2$ constructs the $i$'th block of $n$ bits of $y''$ by simulating $\mathcal{R}_1$ on $\database_1(y^{(i)})$ in the natural way:
 whenever $\mathcal{R}_1$ 
 requests a value $\hat{f}_{\itemset_1}$, $\mathcal{R}_2$
 returns the value $f_{g_i(\itemset_1)}$. $\mathcal{R}_2$ then sets 
 $(y''_{(i-1) \cdot n + 1}, \dots, y''_{i \cdot n})$ to the vector $\hat{y}$ output by $\mathcal{R}_1$.

\medskip
\noindent \textbf{Showing $y''$ is close to} $\text{Enc}(y')$ \textbf{in Hamming distance.}
For each special attribute $i$, let $S_{1, i}$ denote the set of itemsets $\itemset \subseteq [d_2]$ in $\text{Range}(g_i)$ satisfying 
$n \cdot |\hat{f}_{\itemset} - f_{\itemset}(\database_2(y'))| \leq \zeta_1.$ 
Since a $1-\gamma_2$ fraction of the estimates $\hat{f}_{\itemset}$ under $\mu_2$ satisfy 
$n \cdot |\hat{f}_{\itemset} - f_{\itemset}(\database_2(y'))| \leq \zeta_1$,
Markov's inequality implies that 99\% of the $i$'s satisfy $\sum_{\itemset \in S_{1, i}} \mu_2(\itemset)/d_0 \geq 1-100\gamma_2 = 1-\gamma_1$.
Lemma \ref{biglemma} implies that for each such $i$, the $i$th block output by $\mathcal{R}_2$, namely
 $(y''_{(i-1) \cdot n + 1}, \dots, y''_{i \cdot n})$, will have Hamming distance $o(n)$ from the $i$th block of $\Enc(y')$. 
 Hence, $y''$ has fractional Hamming distance at most $.01 + o(1) \leq .02$ from $\Enc(y')$. 
 \end{proof}
 
Rudelson \cite{rudelson} proved the existence of matrices $A_1, \dots, A_{k-1}$ satisfying the conditions of Lemmas \ref{biglemma} and \ref{biglemma2}. 

\begin{lemma}[Rudelson \cite{rudelson}, see also Theorem 5.11 of \cite{de}]
\label{rudelson}
Let $q, k$ be constants. Also, let $\nu \sim \mathbb{R}^{d' \times n}$
be a distribution over
matrices such that every entry of the matrix is an independent and unbiased $\{0, 1\}$ random variable.
Let $A_1,\dots, A_{k-1}$ be i.i.d. copies of random matrices drawn from the distribution $\nu$ and $A$ be the
Hadamard product of $A_1, \dots, A_{k-1}$. Then, provided that $d^{k-1} = o(n \log_{(q)}(n))$,
with probability $1-o(1)$,
the smallest singular value of $A$, denoted by $\sigma_n(A)$, satisfies
$\sigma_n(A) = \Omega\left(\sqrt{d^{k-1}}\right)$.
Also, the range of $A$ is a $(\gamma(q, \ell), n, d^{k-1})$ Euclidean section for some $\gamma(q, \ell) > 0$.
\end{lemma}
 
 Combining Lemmas \ref{rudelson} and Lemma \ref{biglemma2}, we obtain the following lemma.
 
 \begin{lemma} \label{dieforreal}
For any positive constants $q, k$, and any pair $d_0, n > 0$ satisfying $n\log_{(q)}(n) < d_0^{k-1}$, there is a $b=b(d_0, n) \in \Omega(d_0 \cdot n)$ and database generation algorithm $\alg_2$ that takes as input a Boolean vector $y \in \{0, 1\}^{b}$,
and outputs a database $\database_2$ with $n$ rows and $d_2:=(k-1)\cdot d_0 + d_0 = k \cdot d_0$ 
columns such that the following holds.

There exists a constant 
$\gamma_2 = \gamma_2(k, q) > 0$, a distribution $\mu_2$ over $k$-itemsets $\itemset \subseteq [d_2]$, and a reconstruction algorithm $\mathcal{R}_2$ satisfying the following. Let $\zeta_1 = \sqrt{n}/\log_{(q+1)}(n)$. Suppose
$\mathcal{R}_2$ is given $\database_0$ and 
 approximate itemset frequencies $\hat{f}_{\itemset}$ for all $k$-itemsets $\itemset \subseteq [d_2]$. 
 Let $S_2$ denote the set of all $k$-itemsets $\itemset$ satisfying 
$n \cdot |\hat{f}_{\itemset} - f_{\itemset}(\database_2(y'))|  \leq \zeta_1$, and suppose that $\sum_{\itemset \in S_2} \mu_2(\itemset) \geq 1-\gamma_2$.
 Then $\mathcal{R}_2$ outputs $y'$. 
\end{lemma}

\medskip \noindent \textbf{Remark:} Note that Lemma \ref{dieforreal} actually holds for any $\zeta_1 = o(\sqrt{n})$; we choose a particular $\zeta_1$ 
that makes the lemma particularly convenient to apply in our context.

\medskip
For any $\eps>0$, suppose we set $n=1/(\eps^2 \cdot \log_{(q)}(n))$ in the statement of Lemma \ref{dieforreal}. This causes $\zeta_1/n$ to equal $1/(\sqrt{n} \cdot \log_{(q+1)}(n)) = \eps \sqrt{\log_{(q)}(n)}/\log_{(q+1)}(n) > \eps$. Hence, we conclude that an \IFE\ sketch that can answer all $k$-itemset frequency queries with error bounded by $\eps$ provides
sufficiently accurate itemset frequency estimates $\hat{f}_\itemset$ to apply Lemma \ref{dieforreal} with $n=1/(\eps^2 \cdot \log_{(q)}(n))$, and Lemma \ref{lemma:de} follows.

\end{proof}

\end{document}